\documentclass[10pt]{article}
\usepackage{graphicx,epsfig,amsmath,amsthm,amssymb}
\usepackage{enumerate}
\usepackage{multirow}
\usepackage{bbm}
\usepackage{balance}
\usepackage{footnote}
\usepackage{tablefootnote}
\usepackage{xr}
\newcommand{\ex}[1]{\mathbb{E}\left[ #1 \right] }
\newcommand{\norm}[1]{\left\lVert #1 \right\rVert}
\newcommand{\Q}{ {\mathbf Q}}

\newcommand{\A}{ {\mathbf A}}
\newcommand{\s}{ {\mathbf S}}
\newcommand{\UU}{ {\mathbf U}}
\newcommand{\inner}[2]{\langle #1, #2 \rangle}
\newcommand{\lep}[1]{\mathop  \le \limits^{(#1)}}
\newcommand{\gep}[1]{\mathop  \ge \limits^{(#1)}}
\newcommand{\gp}[1]{\mathop  > \limits^{(#1)}}
\newcommand{\ep}[1]{\mathop  = \limits^{(#1)}}
\newcommand{\norms}[1]{\big\lVert #1 \big\rVert}
\newcommand{\normss}[1]{\lVert #1 \rVert}

\newtheorem{definition}{Definition}
\newtheorem{lemma}{Lemma}
\newtheorem{claim}{Claim}
\newtheorem{proposition}{Proposition}

\newtheorem{theorem}{Theorem}

\newtheorem{remark}{Remark}

\theoremstyle{definition}
\newtheorem{condition}{Condition}

\begin{document}






%

\title{Heavy-traffic Delay Optimality in Pull-based  Load Balancing Systems: Necessary and Sufficient Conditions}
\author{Xingyu Zhou\\Department of ECE\\The Ohio State University\\zhou.2055@osu.edu \and Jian Tan\\Department of ECE\\The Ohio State University\\tan.252@osu.edu\\
\and Ness Shroff\\Department of ECE and CSE\\The Ohio State University\\
shroff.11@osu.edu } 
\date{}
\maketitle

\begin{abstract}

In this paper, we consider a load balancing system under a general pull-based policy. In particular, each arrival is randomly dispatched to one of the servers whose queue lengths are below a threshold, if there are any; otherwise, this arrival is randomly dispatched to one of the entire set of servers. We are interested in the fundamental relationship between the threshold and the delay performance of the system in heavy traffic. To this end, we first establish the following necessary condition to guarantee heavy-traffic delay optimality: the threshold will grow to infinity as the exogenous arrival rate approaches the boundary of the capacity region (i.e., the load intensity approaches one) but the growth rate should be slower than a polynomial function of the mean number of tasks in the system. As a special case of this result, we directly show that the delay performance of the popular pull-based policy Join-Idle-Queue (JIQ) lies strictly between that of any heavy-traffic delay optimal policy and that of random routing. We further show that a sufficient condition for heavy-traffic delay optimality is that the threshold grows logarithmically with the mean number of tasks in the system. This result  directly resolves a generalized version of the conjecture by Kelly and Laws.

\end{abstract}



\section{Introduction}

We consider a classical load balancing system that consists of a central dispatcher and $N$ servers, each associated with an infinite buffer queue and a service rate $\mu_n$. The exogenous tasks arrive with rate $\lambda_{\Sigma}$, and upon arrival they must be immediately dispatched to one of the queues. A key to the performance of such a system is the load balancing policy it uses since it directly determines which queue the arriving tasks should join.

To design effective load balancing policies and hence provide good delay performance, it is imperative to develop analytical tools to evaluate the system performance under different load balancing policies. 
Towards that goal, one important line of research has focused on the so-called heavy-traffic regime, where the exogenous arrival rate approaches the boundary of the capacity region, i.e., the heavy-traffic parameter $\epsilon = \sum \mu_n - \lambda_{\Sigma}$ approaches zero. An attractive property of the heavy-traffic regime, as pointed out in~\cite{kelly1993dynamic}, is that \emph{`the important features of good control policies are often displayed in the sharpest relief'.} It has been shown that well-known policies such as Join-Shortest-Queue (JSQ) and Power-of-$d$ can achieve asymptotically optimal delay performance in the heavy-traffic regime~\cite{foschini1978basic,eryilmaz2012asymptotically,chen2012asymptotic,maguluri2014heavy}. Under these two policies, an incoming task is assigned to a server with the shortest queue among $d\ge 2$ servers ($d = N$ for JSQ) sampled uniformly at random.

However, due to the sampling process, the amount of communication overhead is $2d$ per arrival ($d$ for query and $d$ for response), which is undesirable for a large value of $d$, especially in the JSQ policy when $d = N$. More importantly, since the dispatching decision can only be made after collecting the queue length feedback, there exists a non-zero dispatching delay, which contributes to an increase in the response time. To avoid these drawbacks, an alternative approach, often called pull-based load balancing, has received significant recent attention. Instead of actively sending queries to servers and waiting for responses, the dispatcher under a pull-based load balancing scheme passively listens to the reports from the servers. In particular, each server will report its ID to the dispatcher when it satisfies a certain condition (e.g., its queue length drops below a threshold from above). Then, upon task arrival, the dispatcher checks its record. If it is not empty, the dispatcher randomly removes one ID and sends the arrival to the corresponding server; otherwise, it just randomly selects a queue to join. The classical pull-based policy is the Join-Idle-Queue (JIQ) policy investigated in~\cite{lu2011join,stolyar2015pull}, under which the dispatcher maintains a record of IDs of the idle servers (i.e., the reporting threshold is one). JIQ has been shown to enjoy a low message overhead (at most one per arrival), zero dispatching delay, and better delay performance than Power-of-$d$ under medium loads. Nevertheless, under high loads, its delay performance degrades substantially due to the lack of idle servers. This directly suggests that a varying reporting threshold with respect to the load is necessary to guarantee good delay performance in heavy traffic. Motivated by this observation, in a recent work~\cite{zhou2017designing}, the authors propose a specific way to update the reporting threshold in a pull-based policy, which is proven to be heavy-traffic delay optimal, while still enjoying many of the nice features of JIQ.

In this paper, instead of focusing on another specific way of determining the reporting threshold, we step back and work towards answering the following fundamental question: \emph{How would different reporting thresholds affect the (heavy traffic) delay performance of a pull-based policy?} To address this question, we take a systematic approach and summarize the main contributions as follows.
\begin{itemize}
	\item We first present a necessary condition on the reporting threshold for the delay optimality of a pull-based policy in heavy-traffic. In particular, we show that to achieve heavy-traffic delay optimality, the reporting threshold $r$ should grow to infinity as the heavy-traffic parameter $\epsilon$ approaches zero, however, it cannot grow too fast (see Theorem \ref{thm:constant}). An important corollary of Theorem \ref{thm:constant} is that the delay performance of the JIQ policy (i.e., constant threshold $r = 1$) in heavy traffic lies \emph{strictly} between that of any heavy-traffic delay optimal policies (e.g., JSQ) and that of random routing. This result is somewhat counter-intuitive, since at first glance one may guess that JIQ would degenerate to random routing in heavy traffic since there are hardly any idle servers in the system. However, it turns out that it is not true, and allows us to get a sharp characterization of the JIQ policy in heavy traffic. 
	\item We then establish a sufficient condition on the reporting threshold for heavy-traffic delay optimality of pull-based policies. Specifically, we show that a logarithmic growth rate of the reporting threshold with respect to the mean number of tasks in the system is sufficient to guarantee the steady-state delay optimality in heavy traffic (see Theorem \ref{thm:log}). This result directly resolves a conjecture by Kelly and Laws in~\cite{kelly1993dynamic}. In particular, the authors in~\cite{kelly1993dynamic} consider a two-server system with Poisson arrivals and exponential service under a varying reporting threshold. They conjecture that as long as the threshold is greater than a specified constant times the logarithm of the mean number of tasks in the system, then asymptotic delay optimality holds in heavy traffic. Thus, our result not only resolves the conjecture but generalizes it to any fixed finite number of servers with general arrival and service distributions. It is also worthing noting that the asymptotic delay optimality achieved in our paper is in steady-state while delay optimality in~\cite{kelly1993dynamic} holds only for a finite time interval.
	\item The techniques introduced in this paper may be of independent interest for the analysis of general load balancing  policies. More precisely, the key to establishing heavy-traffic delay optimality in this paper is a notion of state-space collapse, which is different from the state-space collapse result often adopted in previous works. As a result, it requires us to develop a new Lyapunov function to conduct the drift analysis. More importantly, due to this new type of state-space collapse, we have to devise a new approach to relate the state-space collapse result to the final heavy-traffic delay optimality. 
\end{itemize}

\subsection{Related Work}

The investigation of queueing delay in heavy traffic with dynamic routing dates back to~\cite{foschini1978basic}, in which the authors considered a two-server system under the JSQ policy, and they showed that the two separate servers under JSQ act as a pooled resource in heavy traffic via diffusion approximations. Since then, the methodology of diffusion approximations has been adopted in a number of works on parallel queues~\cite{chen2012asymptotic,reiman1984some,hanqin1989heavy,williams1998diffusion,harrison1998heavy,bramson1998state}. For example, the author in~\cite{reiman1984some} generalized the results in~\cite{foschini1978basic} to the case of renewal arrivals and general service times. The functional central limit theorems for the JSQ policy in a load balancing system with multiple servers was derived in~\cite{hanqin1989heavy}. In~\cite{chen2012asymptotic}, the Power-of-d policy was shown to have the same diffusion limit as JSQ in the heavy-traffic limit. Much of the works based on the diffusion approximation method rely on showing that a scaled version of queue lengths converges to a regulated Brownian motion. This result typically leads to a sample-path optimality in a finite time interval. However, showing the convergence to the steady-state distribution requires the additional validation of the interchange of limits, which is often not taken (some exceptions include~\cite{gamarnik2006validity,budhiraja2009stationary}, in which the authors proved an interchange of limit argument for generalized Jackson networks with a fixed routing matrix). Motivated by this, the authors in~\cite{eryilmaz2012asymptotically} proposed a Lyapunov drift-based approach, which is able to establish steady-state heavy-traffic optimality of the load balancing policy JSQ and scheduling policy MaxWeight. One of the main features of this framework is that it is able to avoid the interchange-of-limits issue by directly working on the stationary distribution. This approach has been utilized to show steady-state heavy-traffic delay optimality of Power-of-$d$ in~\cite{maguluri2014heavy}. Moreover, based on this approach, it has  been shown in~\cite{wang2016maptask} that a joint JSQ and MaxWeight policy is heavy-traffic delay optimal for MapReduce clusters.

As discussed in the introduction, while JSQ and Power-of-$d$ enjoy heavy-traffic delay optimality, they both have non-zero dispatching delay, and a relatively high message overhead. Motivated by this, a pull-based design of load balancing policies has gained significant recent popularity. The main feature of pull-based load balancing is the introduction of local memory at the dispatcher, which maintains a record of servers satisfying a pre-defined condition (e.g., its queue length is below a threshold in most cases). The dispatching decision is made purely based on the local memory: if it is nonempty, randomly choosing a server in memory to join; otherwise, randomly choosing a server from all the servers. For instance, one illustrative example is the JIQ policy proposed and studied in~\cite{lu2011join,stolyar2015pull}, under which the local memory maintains all the idle servers. As a result, the arrival is always dispatched to one of the idle servers if there are any; otherwise, it is dispatched randomly. It has been shown that JIQ has a low message overhead (at most one per arrival), zero dispatching delay, and better performance compared to Power-of-$2$ in medium loads. Nevertheless, since only the idle servers are stored in memory, when the loads become high, its performance degrades substantially because the memory is empty and hence random routing is adopted most of the time. Therefore, this directly suggests that a varying threshold is necessary to guarantee good performance in heavy traffic for a pull-based policy. 

To this end, in a recent work~\cite{zhou2017designing}, the authors successfully propose a pull-based policy with a varying threshold, which is proven to be heavy-traffic delay optimal in steady state while keeping the nice features of JIQ. This naturally raises the question about the fundamental relationship between the choice of the threshold and the delay performance, which is the main focus of this paper. In particular, our work is mainly motivated by the seminal paper~\cite{kelly1993dynamic}, in which Kelly and Laws give a conjecture regarding the choice of the threshold that is able to guarantee delay optimality in heavy traffic. More precisely, they consider a two-server system with Poisson arrivals and exponential service. The arrival is dispatched randomly, except when one queue is below the threshold $r$ and the other is above, in which case the arrival is dispatched to the shorter one. Note that this dynamic policy can be exactly implemented by a pull-based load balancing scheme with a threshold $r$. Kelly and Laws conjecture that as long as the threshold $r$ is greater than a specific constant times the logarithm of the mean number of tasks in the system, then the sum queue lengths process under this threshold policy has the same diffusion limit as that under JSQ. Therefore, the logarithmic growth rate result in our sufficient conditions (see Theorem \ref{thm:log}) not only directly resolves the conjecture in~\cite{kelly1993dynamic}, but generalizes it to systems with any fixed finite number of servers as well as general arrival and service distributions. Moreover, the diffusion limit result conjectured in~\cite{kelly1993dynamic} only gives the optimality in a finite time interval while our heavy traffic optimality result obtained by Lyapunov drift-based approach is in steady state.

It is also worth noting that a logarithmic growth in the threshold is not a coincidence, and has been found in a wide range of scenarios. For example, the authors in~\cite{teh2002critical} consider an asymmetric threshold policy for a two-server case. 
In that setting, only one server has a threshold $r$ (say server 2). The arrivals are always dispatched to server 1 unless the queue length of server 2 is less than the threshold, in which case the arrival is sent to server 2. 
One of the main contributions in~\cite{teh2002critical} is that a logarithmic growth rate of $r$ is sufficient to guarantee that this threshold policy achieves the same diffusion limit as that under JSQ in heavy traffic. This result can be seen as a first attempt to resolve the conjecture in~\cite{kelly1993dynamic} with a simpler model. In particular, since there is only one threshold in~\cite{teh2002critical}, the network can be characterized by a one-dimensional reflected Brownian motion in heavy traffic. In contrast, the limit process in~\cite{kelly1993dynamic} is a two-dimensional Brownian motion, which is harder to rigorously prove optimality. Besides dynamic routing, a logarithmic growth rate of the threshold also critically affects the performance of scheduling policies in~\cite{harrison1998heavy,bell2001dynamic}. Both authors considered a system of two parallel servers with dedicated arrivals to each of the queues. One server can only process tasks in its own queue, while a `super-server' can process tasks from both queues. A threshold policy is proposed in which the `super-server' processes tasks from its own queue when the other server's queue length is below a threshold, and otherwise the `super-server' processes the tasks from the other queue. This policy can be viewed as the scheduling counterpart of the asymmetric routing policy considered in~\cite{teh2002critical}. In a `discrete review' setting, the author in~\cite{harrison1998heavy} proved that a sufficient condition for the asymptotic optimality of this threshold policy is that the threshold must grow as a constant times the average number of tasks in the system. The same result was generalized to a `continuous review' setting with more general arrival and service distributions in~\cite{bell2001dynamic}. As in the paper by Kelly and Laws~\cite{kelly1993dynamic}, the asymptotic optimality in~\cite{teh2002critical,harrison1998heavy,bell2001dynamic} holds in a finite time interval since the convergence to the stationary distribution is not validated for the diffusion approximations. Considering the similarity between the scheduling policies in~\cite{harrison1998heavy,bell2001dynamic} and the routing policy in~\cite{teh2002critical}, our approach developed in this paper might be applied to establish heavy-traffic delay optimality in steady state for dynamic scheduling policies as well.
We shall finally point out that the heavy-traffic regime considered in this paper and all the aforementioned papers assumes that the number of servers is a constant, which is different from the Halfin-Whitt heavy-traffic regime (also known as many-server heavy-traffic regime or quality-and-efficiency-driven regime)~\cite{halfin1981heavy}. In the latter regime, the heavy-traffic parameter $\epsilon$ approaches zero and the number of servers $N$ goes to infinity at the same time~\cite{armony2005dynamic,gurvich2009queue,dai2011state,mukherjee2016universality}. For example, it has been shown that, on any finite time interval, the limiting process under the JIQ policy is indistinguishable from that under the JSQ policy in the Halfin-Whitt heavy-traffic regime~\cite{mukherjee2016universality}. 
In contrast, in the conventional heavy-traffic regime considered in this paper, its delay performance is strictly between that of JSQ and random routing as shown by Theorem~\ref{thm:constant}.

\subsection{Notations}

The dot product in $\mathbb{R}^N$ is denoted by $\inner{\mathbf{x} }{\mathbf{y} } \triangleq \sum_{n=1}^N x_ny_n$. For any $\mathbf{x} \in \mathbb{R}^N$, the $l_1$ norm is denoted by $\norm{\mathbf{x}}_1 \triangleq \sum_{n=1}^N |x_n|$ and $l_2$ norm is denoted by $\norm{\mathbf{x}} \triangleq \sqrt{\inner{\mathbf{x}}{\mathbf{x}}}$. In general, the $l_r$ norm is denoted by $\norm{\mathbf{x}}_r \triangleq (\sum_{n=1}^N |x_n|^r)^{1/r}$. Let $\mathcal{N}$ denote the set $\{1,2,\ldots,N\}$.

\section{System Model and Preliminaries}
This section first describes the system model and assumptions considered in this paper. Then, several necessary preliminaries are presented.
\subsection{System model}
We consider a discrete-time load balancing system consisting of a central dispatcher and $N$ servers. Each server maintains an infinite capacity FIFO queue. At the central dispatcher, there is also a local memory denoted as $m(t)$, through which the dispatcher can have limited information about the system. In each time-slot, the central dispatcher routes the new incoming tasks to one of the servers, immediately upon arrival as in~\cite{eryilmaz2012asymptotically,maguluri2014heavy,wang2016maptask,xie2015priority,xie2016scheduling,zhou2017designing}. Once a task joins a queue, it will remain in that queue until its service is completed. Each server is assumed to be work conserving: a server is idle if and only if its corresponding queue is empty.

\subsubsection{Arrival and Service} Let $A_{\Sigma}(t)$ denote the number of exogenous tasks that arrive at the beginning of time-slot $t$. We assume that $A_{\Sigma}(t)$ is an integer-valued random variable, which is \emph{i.i.d.} across time-slots. The mean and variance of $A_{\Sigma}(t)$ are denoted by $\lambda_{\Sigma}$ and $\sigma_{\Sigma}^2$, respectively. We further assume that there is a positive probability for $A_{\Sigma}(t)$ to be zero.
Let $S_n(t)$ denote the amount of service that server $n$ offers for queue $n$ in time-slot $t$. Note that this is not necessarily equal to the number of tasks that leaves the queue because the queue may be empty. We assume that $S_n(t)$ is an integer-valued random variable, which is \emph{i.i.d.} across time-slots. We also assume that $S_n(t)$ is independent across different servers as well as the arrival process. The mean and variance of $S_n(t)$ are denoted as $\mu_n$ and $\nu_n^2$, respectively. Let $\mu_{\Sigma} \triangleq \Sigma_{n=1}^N \mu_n$ and $\nu_{\Sigma}^2 \triangleq \Sigma_{n=1}^N \nu_n^2$ denote the mean and variance of the hypothetical total service process $S_{\Sigma}(t) \triangleq \sum_{n=1}^N S_n(t)$. 
To illustrate the key ideas behind the results, we first assume that both the arrival and service processes have a finite support, i.e., $A_{\Sigma}(t) \le A_{max} < \infty$ and $S_n(t) \le S_{max} <\infty$ for all $t$ and $n$. However, the main results still hold when the support is infinite, as discussed in Section \ref{sec:general}.

\subsubsection{Queue Dynamics} Let $Q_n(t)$ be the queue length of server $n$ at the beginning of time slot $t$. 
Let $A_n(t)$ denote the number of tasks routed to queue $n$ at the beginning of time-slot $t$ according to the dispatching decision. 
Then the evolution of the length of queue $n$ is given by 
\begin{equation}
	\label{eq:Qdynamic}
	Q_n(t+1) = Q_n(t) + A_n(t) - S_n(t) + U_n(t), n = 1,2,\ldots, N,
\end{equation}
where $U_n(t) = \max\{S_n(t)-Q_n(t)-A_n(t),0\}$ is the unused service due to an empty queue.

\subsection{Preliminaries}
In this paper, we are interested in a general pull-based policy formally defined as follows. In words, under this policy, the arrival is randomly dispatched to one of the servers whose queue lengths are below a threshold $r$, if there are any; Otherwise, it is dispatched to one of $N$ queues randomly.
\begin{definition}
	Join-Below-Threshold (JBT) policy is composed of the following components:
	\begin{enumerate}[(a)]
		\item Each server $n$ sends its ID to the dispatcher when its queue length is below the threshold $r$ for the first time.
		\item Upon a new arrival, the dispatcher checks the available IDs in the memory. If they exist, it removes one uniformly at random, and sends all the new arrivals to the corresponding server. Otherwise, all the new arrivals are dispatched uniformly at random to one of the servers in the system.
		\item For the case of heterogeneous servers, in (a) each server also sends its $\mu_n$ to the dispatcher and in (b) instead of choosing the ID uniformly at random, the dispatcher selects the ID in proportion to the service rate, that is, if the ID of server $i$ is in $m(t)$, the probability for server $i$ to be chosen is $\mu_i/\sum_{j \in m(t)} \mu_j$.
	\end{enumerate}
\end{definition}
\begin{remark}
	It is easy to see that JIQ is a special case of JBT with $r = 1$. 
\end{remark}
 
The considered load balancing system under JBT can be modeled as a discrete-time Markov chain $\{Z(t) = (\Q(t),m(t)), t\ge 0\}$ with state space $\mathcal{Z}$, using the queue length vector $\Q(t)$ together with the memory state $m(t)$. We consider a set of load balancing systems $\{Z^{(\epsilon)}(t), t\ge 0\}$ parameterized by $\epsilon$ such that the mean arrival rate of the exogenous arrival process $\{A_{\Sigma}^{(\epsilon)}(t), t\ge 0\}$ is $\lambda_{\Sigma}^{(\epsilon)} = \mu_\Sigma - \epsilon$. Note that the parameter $\epsilon$ characterizes the distance between the arrival rate and the boundary of the capacity region. We are interested in the throughput performance and more importantly the steady-state delay performance in the heavy-traffic regime under the JBT policy.

Recall that a load balancing system is stable if the Markov chain $\{Z(t), t\ge 0\}$ is positive recurrent, and $\overline{Z} = \{\overline{\Q}, \overline{m}\}$ denotes the random vector whose distribution is the same as the steady-state distribution of $\{Z(t), t\ge 0\}$. We have the following definition.
{
\begin{definition}[Throughput Optimality]
	A load balancing policy is said to be throughput optimal if for any arrival rate within the capacity region, i.e., for any $\epsilon > 0$, the system is positive recurrence and all the moments of $\big\lVert{\overline{\Q}^{(\epsilon)}}\big\rVert$ are finite.
\end{definition}

Note that this is a stronger definition of throughput optimality than that in~\cite{wang2016maptask,xie2016scheduling,zhou2017designing}, because besides the positive recurrence, it also requires all the moments to be finite in steady state for any arrival rate within the capacity region.}




To characterize the steady-state average delay performance in the heavy-traffic regime when $\epsilon$ approaches zero, by Little's law, it is sufficient to focus on the summation of all the queue lengths. First, recall the following fundamental lower bound on the expected sum queue lengths in a load balancing system under any throughput optimal policy \cite{eryilmaz2012asymptotically}.

\begin{lemma}
\label{lem:lower_bound}
    Given any throughput optimal policy and assuming that $(\sigma_{\Sigma}^{(\epsilon)})^2$ converges to a constant $\sigma_{\Sigma}^2$ as $\epsilon$ decreases to zero, then 
	\begin{equation}
	\label{eq:lower_bound}
		\liminf_{\epsilon \downarrow 0} \epsilon \ex{\sum_{n=1}^N \overline{Q}_n^{(\epsilon)} } \ge \frac{\zeta}{2},
	\end{equation}
	where $\zeta \triangleq \sigma_{\Sigma}^2 + \nu_{\Sigma}^2$.
\end{lemma}

The right-hand-side of Eq. \eqref{eq:lower_bound} is the heavy-traffic limit of a hypothetic single-server system with arrival process $A_\Sigma^{(\epsilon)}(t)$ and service process $\sum_n^N S_n(t)$ for all $t\ge0$. This hypothetical single-server queueing system is often called the \textit{resource-pooled system}. Since a task cannot be moved from one queue to another in the load balancing system, it is easy to see that the expected sum queue lengths of the load balancing system is larger than the expected queue length in the resource-pooled system. However, under a certain load balancing policy, the lower bound in Eq. \eqref{eq:lower_bound} can actually be attained in the heavy-traffic limit and hence based on Little's law this policy achieves the minimum average delay of the system in steady-state. This directly motivates the following definition of steady-state heavy-traffic delay optimality as in \cite{eryilmaz2012asymptotically,maguluri2014heavy,wang2016maptask,xie2015priority,xie2016scheduling,zhou2017designing}.

\begin{definition}[Heavy-traffic Delay Optimality in Steady-state]
	A load balancing scheme is said to be heavy-traffic delay optimal in steady-state if the steady-state queue length vector $\overline{\Q}^{(\epsilon)}$ satisfies 
	\begin{equation*}
		\limsup_{\epsilon \downarrow 0} \epsilon \ex{\sum_{n=1}^N \overline{Q}_n^{(\epsilon)} } \le \frac{\zeta}{2},
	\end{equation*}
	where $\zeta$ is defined in Lemma \ref{lem:lower_bound}.
\end{definition}


In the analysis of the delay performance of JBT, the following region $\mathcal{R}^{(r)}$ in $\mathbb{R}^N$ plays an instrumental role by the virtue of the JBT policy.
\begin{align}
\label{eq:def_R}
	\mathcal{R}^{(r)} &= \mathcal{R}_{l}^{(r)} \cup \mathcal{R}_u^{(r)},
\end{align}
where $r \ge 1$ and
\begin{align*}
&\mathcal{R}_{l}^{(r)} \triangleq \left\{ \mathbf{x} \in \mathbb{R}^N_+ : x_n \le r \text{ for all } n \in \mathcal{N} \right\} \\
&\mathcal{R}_u^{(r)} \triangleq \left\{ \mathbf{x} \in \mathbb{R}^N_+ : x_n \ge r  \text{ for all } n \in \mathcal{N} \right\}.
\end{align*}
By the definition of the JBT policy, we have that whenever the queue lengths vector is within the region $\mathcal{R}^{(r)}$, then JBT reduces to (proportionally) random routing. On the other hand, when the queue lengths vector is outside the region $\mathcal{R}^{(r)}$, shorter queues are preferred over longer queues.
\section{Main Results}
In this section, we present both necessary and sufficient conditions on the threshold $r$ for the JBT policy to be heavy-traffic delay optimal in steady-state. We first establish throughput optimality of the JBT policy, which serves as a basis for the analysis of heavy-traffic delay optimality.

\subsection{Throughput optimality}

We first prove the following result, which establishes that a load balancing system under the JBT policy is stable with bounded moments on the queue lengths for any threshold $r\ge 1$. 
\begin{lemma}
\label{thm:throughput}
JBT is throughput optimal with the $p$-th moment of $\norms{\overline{\Q}^{(\epsilon)}}$ being $O(1/\epsilon^p)$ for any threshold $r\ge 1$ and integer $p\ge 1$.
\end{lemma}

\begin{proof}
	See Appendix \ref{sec:appex_proof_throughput}
\end{proof}

Besides throughput optimality, another important aspect of this lemma is that it serves as the basis for the discussions on heavy-traffic delay optimality in the following sections. This is because, firstly, a load balancing policy that cannot stabilize the system is incapable of being heavy-traffic delay optimal at all. Second, the bounded moments result allows us to set the mean drift of Lyapunov functions concerning queue lengths to be zero in steady state, which plays a pivotal part in the framework of Lyapunov drift-based heavy-traffic analysis.

\subsection{Necessary condition}
In this section, we show that a necessary condition for the JBT policy to achieve heavy-traffic delay optimality is that the threshold $r$ should grow to infinity as the heavy-traffic parameter $\epsilon$ approaches zero. However, as we show it cannot grow too fast. Formally, it is presented in the following theorem.

\begin{theorem}
\label{thm:constant}
Consider a  load balancing system with homogeneous servers under the JBT policy. 
\begin{enumerate}
	\item Suppose the threshold $r$ is any constant in $[1,\infty)$, then we have 
\begin{align}
\label{eq:liminf}
	 \liminf_{\epsilon \downarrow 0} \epsilon \ex{\sum_{n=1}^N \overline{Q}_n^{(\epsilon)} }  > \frac{\zeta}{2} 
\end{align}
and 
\begin{align}
\label{eq:limsup}
	\limsup_{\epsilon \downarrow 0} \epsilon \ex{\sum_{n=1}^N \overline{Q}_n^{(\epsilon)} } < \lim_{\epsilon \downarrow 0} \epsilon \ex{\sum_{n=1}^N \overline{Q}_{n,\text{Rand} }^{(\epsilon)} },
\end{align}
where $\overline{\Q}_{\text{Rand}}^{(\epsilon)}$ is the steady-state vector under random routing policy.
\item Suppose the threshold $r^{(\epsilon)} = (1/\epsilon)^{1+\alpha}$ for any constant $\alpha > 0$, then we have 
\begin{align}
\label{eq:necessary_upper}
	\lim_{\epsilon \downarrow 0} \epsilon \ex{\sum_{n=1}^N \overline{Q}_n^{(\epsilon)} } = \lim_{\epsilon \downarrow 0} \epsilon \ex{\sum_{n=1}^N \overline{Q}_{n,\text{Rand} }^{(\epsilon)} }.
\end{align}
\end{enumerate}

\end{theorem}
\begin{proof}
	See Section \ref{sec:proof_thm_constant}
\end{proof}

\begin{figure}[t]
	\graphicspath{{./Figures/}}
	\centering
	\includegraphics[width=3.5in]{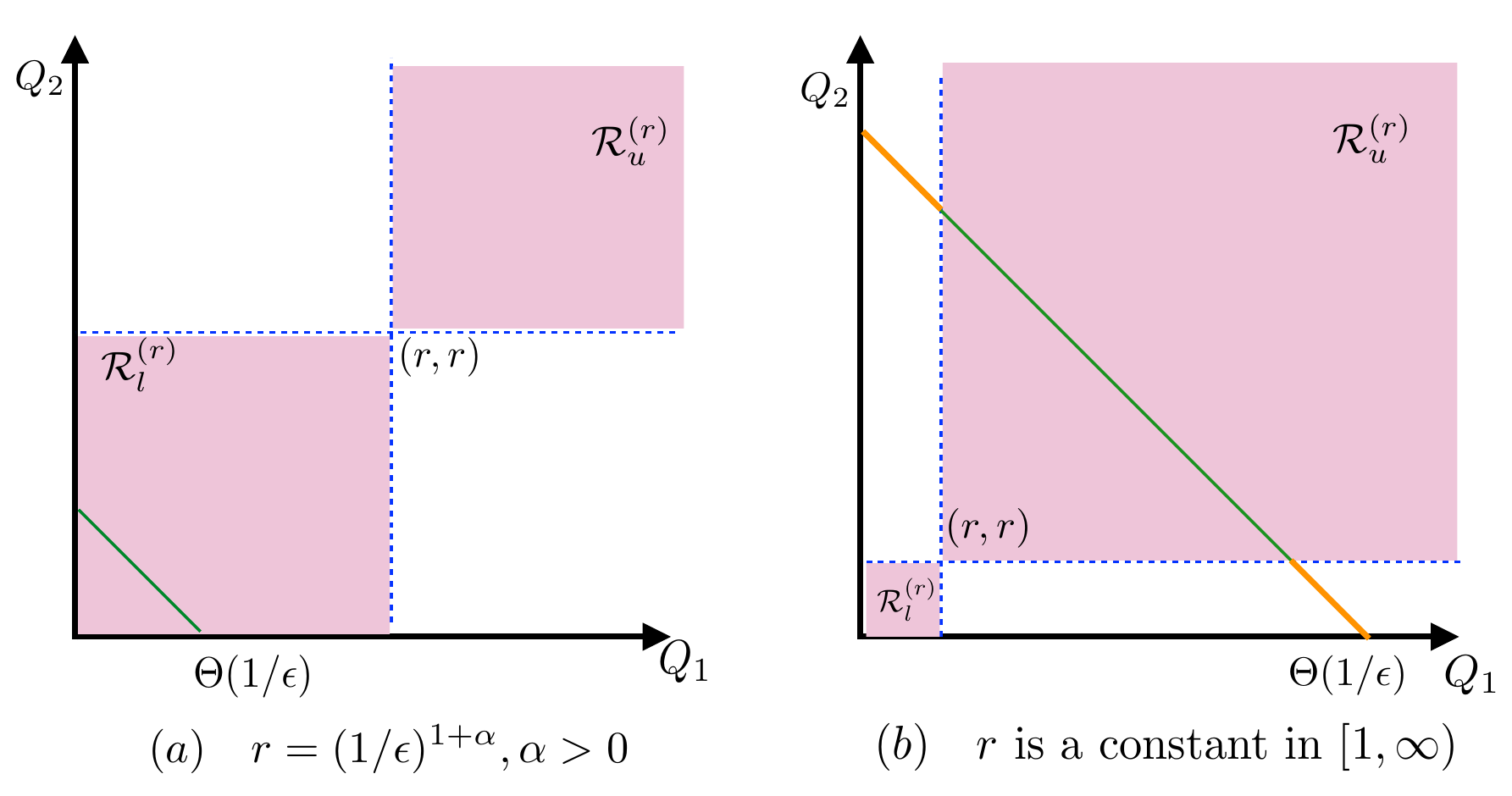}
	\caption{Geometric illustrations of the necessary condition.}\label{fig:necessary}
\end{figure}

Now, we will present the high-level intuitions behind the necessary condition with the illustration in Fig. \ref{fig:necessary}. These intuitions can not only facilitate understanding of the results, but also motivates the sufficient condition in the next section.

To start with, let us consider case (2) when $r^{(\epsilon)} = (1/\epsilon)^{1+\alpha}$ for any $\alpha > 0$. In this case, all the queue lengths are below the threshold $r$ for high loads since the sum queue lengths in the system is only on the order of $1/\epsilon$. As a result, in case (2), the JBT policy completely degenerates to random routing, which is not heavy-traffic delay optimal~\cite{foschini1978basic}. An illustration of case (2) for a two-server system is presented in Fig. 1(a).

Then, we turn to case (1) for which the threshold is a constant. In particular, combing Eqs. \eqref{eq:liminf} and \eqref{eq:limsup} yields that the delay performance of JBT under any constant $r$ in heavy-traffic lies \emph{strictly} between that of a heavy-traffic delay optimal policy (e.g., JSQ) and that of random routing. This reveals an interesting and kind of counter-intuitive insight about the JBT policy under a constant threshold. For example, consider the special case $r = 1$, i.e., the JIQ policy. At first glance, one might expect that the delay performance of JIQ would downgrade to that of random routing in the heavy-traffic limit, since in this case there are hardly any idle servers, and hence the dispatcher under JIQ would just randomly choose one server when allocating arrivals, as in random routing. However, it turns out that this is not true as shown in Eq. \eqref{eq:limsup}. That is, the performance of JIQ is still strictly better than that of random routing even in the heavy-traffic limit. This demonstrates that JIQ is able to achieve \emph{partial} resource pooling due to the fact that it adopts queue lengths information to prefer shorter queues whenever possible. To see this, note that by positive recurrence, there always exists some time when the queue length vector is outside the region $\mathcal{R}^{(r)}$ and hence shorter queues are preferred (i.e., the orange line in Fig. 1(b)), even though it is much less than the time within the region $\mathcal{R}^{(r)}$ (i.e., the green line in Fig. 1(b)). This is totally different from the case in Fig. 1(a) in which the queue-length state always completely remains within the $\mathcal{R}^{(r)}$ for high loads, and hence JBT would downgrade to random routing in the limit.

On the other hand, to explain the liminf result in Eq. \eqref{eq:liminf}, we will utilize the following result. That is, the necessary (and sufficient) condition for the JBT policy to be heavy-traffic delay optimal is given by
\begin{align}
\label{eq:necessary_main}
			\lim_{\epsilon \downarrow 0}\ex{\big\lVert\overline{\Q}^{(\epsilon)}(t+1) \big\rVert_1 \big\lVert\overline{\UU}^{(\epsilon)}(t) \big\rVert_1} = 0.
\end{align}
This is a direct application of the results in~\cite{zhou2018flexible}.
Note that since $Q_n(t+1)U_n(t) = 0$, the above condition basically means that the key for JBT to be heavy-traffic delay optimal is that it should guarantee that no server is idling while other servers are busy with high loads. In the case when $r$ is a constant, the event that one queue is zero while others with high loads (denoted by $E_{\text{bad}}$) happens with a non-negligible probability since the axes are close to the region $\mathcal{R}_u^{(r)}$. As a result, the left-hand side of Eq. \eqref{eq:necessary_main} is strictly positive, and hence JBT is not heavy-traffic delay optimal for a constant $r$. The intuition that we should guarantee that the event $E_{\text{bad}}$ occurs very rarely in heavy-traffic also motivates  
our sufficient condition in the next section where we let the threshold $r$ grows in a certain rate to guarantee that the axes are far away from the region $\mathcal{R}_u^{(r)}$.

\begin{remark}
	It is worth noting that in~\cite{zhou2017designing}, a similar result as Eq. \eqref{eq:liminf} has been established for the JIQ policy (i.e., the special case $r = 1$ of JBT) in a two-server system under the constraints that the service processes are constant and the variance of arrival process should be larger than a particular value. Thus, our contribution is to generalize the result in~\cite{zhou2017designing} to any constant $r\ge1$ and any finite number of servers without the constraints on service and arrival process as required in~\cite{zhou2017designing}. More importantly, we provide new results given by Eqs. \eqref{eq:limsup} and \eqref{eq:necessary_upper}, which give us a sharper understanding of general pull-based policies.
\end{remark}

\subsection{Sufficient condition}
In this section, we now investigate the sufficient condition. In particular, we show that if the threshold in JBT grows at a logarithmic rate with respect to the average sum queue lengths, i.e., $r^{(\epsilon)} \ge K \log(1/\epsilon)$ for some specified constant $K$, then the JBT policy is heavy-traffic delay optimal in steady state, which is formally presented in the following theorem.


\begin{theorem}
\label{thm:log}
	Consider a load balancing system under the JBT policy. Suppose that the threshold $r$ satisfies $r^{(\epsilon)} \ge K\log(1/\epsilon)$ and $r^{(\epsilon)} = o(1/\epsilon)$, where the constant $K = 2(1+\alpha)/\theta^*$ for any $\alpha>0$ and $\theta^*$ is the constant in Eq. \eqref{eq:collpase}, then JBT is heavy-traffic delay optimal in steady state.
\end{theorem}
\begin{proof}
	See Section \ref{sec:proof_theorem_log}
\end{proof}

The main contributions of this result can be summarized as follows. First, it directly resolves and generalizes a conjecture in~\cite{kelly1993dynamic}. More precisely, the authors in~\cite{kelly1993dynamic} consider a two-server system with Poisson arrivals and exponential service under a threshold policy that has the same implementation as JBT, and conjecture that as long as the threshold is greater than a specified constant times $\log(1/\epsilon)$, the heavy-traffic asymptotic optimality of the threshold routing strategy holds. Thus, our result resolves this conjecture and also generalizes it to any finite number of servers case with general arrival and service distributions.  More importantly, the asymptotic optimality defined in~\cite{kelly1993dynamic} holds only for a finite time interval since the convergence to steady-state distribution is not touched. In contrast, our result directly gives the steady-state characterization of the delay optimality in heavy-traffic of the JBT policy.

The key step in establishing the sufficient condition in Theorem \ref{thm:log} is the notion of state-space collapse. In words, it says that in heavy traffic the system state under the JBT policy would concentrate around the region $\mathcal{R}^{(r)}$ as defined Eq. \eqref{eq:def_R}.
To that end, we need the following property of the distance to the region $\mathcal{R}^{(r)}$.
The distance of a point $\mathbf{x}$ to the region $\mathcal{R}^{(r)}$ is related to the distances to the regions $\mathcal{R}_{l}^{(r)}$ and $\mathcal{R}_u^{(r)}$ as follows.
\begin{align}
\label{eq:distoregion}
	d_{\mathcal{R}^{(r)}}(\mathbf{x}) = \min\left(d_{\mathcal{R}_l^{(r)}}(\mathbf{x}), d_{\mathcal{R}_u^{(r)}}(\mathbf{x})\right),
\end{align}
where the distance of a point $\mathbf{x}$ to a set $\mathcal{A}$ in $\mathbb{R}^N$ is defined as
\begin{align*}
	d_{\mathcal{A}}(\mathbf{x})\triangleq \inf_{\mathbf{y} \in \mathcal{A}} \left\{ \norm{\mathbf{x} - \mathbf{y}}\right\}.
\end{align*}
This equality \eqref{eq:distoregion} can be established by contradiction. Suppose that 
\begin{align*}
	\min\left(d_{\mathcal{R}_l^{(r)}}(\mathbf{x}), d_{\mathcal{R}_u^{(r)}}(\mathbf{x})\right) = d_{\mathcal{R}^{(r)}}(\mathbf{x}) + \alpha
\end{align*}
for some $\alpha > 0$, then there exists a $\mathbf{y}^* \in \mathcal{R}^{(r)}$ such that 
\begin{align*}
	d_{\mathcal{R}^{(r)}}(\mathbf{x}) \le \norm{\mathbf{x} - \mathbf{y}^*} < \min\left(d_{\mathcal{R}_l^{(r)}}(\mathbf{x}), d_{\mathcal{R}_u^{(r)}}(\mathbf{x})\right).
\end{align*}
However, since $\mathbf{y}^* \in \mathcal{R}^{(r)} = \mathcal{R}_{l}^{(r)} \cup  \mathcal{R}_u^{(r)}$, this leads to a contradiction to the right-hand side of the inequality above.

We say that the system state concentrates around the region $\mathcal{R}^{(r)}$ if all the moments of the distance $d_{\mathcal{R}^{(r)}}(\overline{\Q})$ are upper bounded by constants. Formally, we have the following definition.
\begin{definition}[State-space collapse to $\mathcal{R}^{(r)}$]
	Suppose that the system process converges in distribution to a steady-state random vector $\overline{\Q}^{(\epsilon)}$. Then, we say that the state-space of a load balancing system collapses to the region $\mathcal{R}^{(r)}$ if there exist some positive constants $\epsilon_0$, $\theta^*$ and $C^*$ such that for all $\epsilon \in (0, \epsilon_0)$
	\begin{align}
	\label{eq:collpase}
		\ex{ e^{\theta^* d_{\mathcal{R}^{(r)}}\big(\overline{\Q}^{(\epsilon)}\big)} } \le C^*,
	\end{align}
	where both $\theta^*$ and $C^*$ are independent of $\epsilon$.
\end{definition}
Note that this notion of state-space collapse is different from previous works, as will be explained later.
 For any constant threshold $r$, Eq. \eqref{eq:collpase} trivially holds since the distance to the region $\mathcal{R}^{(r)}$ is always bounded by a constant. Thus, in the following we only consider the interesting case when $r$ grows to infinity, which is also required by the necessary condition in Theorem \ref{thm:constant}. In this case, we have the following result regarding state space collapse of the JBT policy, which plays a key role in the proof of Theorem \ref{thm:log}. 
\begin{proposition}
\label{thm:collapse}
	Consider a load balancing system under the JBT policy. Suppose that the threshold satisfies $\lim_{\epsilon \downarrow 0} r^{(\epsilon)} = \infty$, then the system state-space collapses to the region $\mathcal{R}^{(r)}$. 
\end{proposition}
\begin{proof}
	See Section \ref{sec:proof_thm_collapse}
\end{proof}

\begin{remark}
	It should be noted that besides being a key step in proving the sufficient conditions in Theorem \ref{thm:log}, Proposition \ref{thm:collapse} has its own contributions. (i) First, the region of state-space collapse in this paper, i.e., $\mathcal{R}^{(r)}$ is not a single dimensional line as in~\cite{eryilmaz2012asymptotically,maguluri2014heavy,wang2016maptask,xie2015priority,xie2016scheduling,zhou2017designing}, nor a multi-dimensional convex cone as in~\cite{maguluri2016heavy,maguluri2018optimal,zhou2018flexible,wang2018heavy}. This not only brings new challenges in proving state-space collapse itself, but also requires new methods to relate the collapse result to heavy-traffic delay optimality. More specifically, on the one hand, in order to prove state-space collapse result, we need to handle the non-convexity of $\mathcal{R}^{(r)}$ by choosing the minimum of two distances as the Lyapunov function. The techniques suggested in~\cite{zhou2018flexible} to handle the non-convex region cannot apply here since the region $\mathcal{R}^{(r)}$ cannot be covered by the cone define in~\cite{zhou2018flexible}. On the other hand, in order to utilize the state-space collapse result to conclude heavy-traffic delay optimality, the conventional decompositions of parallel and perpendicular components of the queue length vector $\Q$ would not work. Instead, we need to carefully divide the system state and then apply Chernoff bound on the random variable $d_{\mathcal{R}^{(r)}}\big(\overline{\Q}^{(\epsilon)}\big)$, which is possible by the state-space collapse result in Eq. \eqref{eq:collpase}. (ii) Second, the upper bound result in Eq. \eqref{eq:collpase} holds even the system is not at the heavy-traffic limit, and hence it is of independent interest for analyzing the system performance in the pre-limit regime, especially when combined with optimization techniques.
\end{remark}

\begin{figure}[t]
	\graphicspath{{./Figures/}}
	\centering
	\includegraphics[width=3.5in]{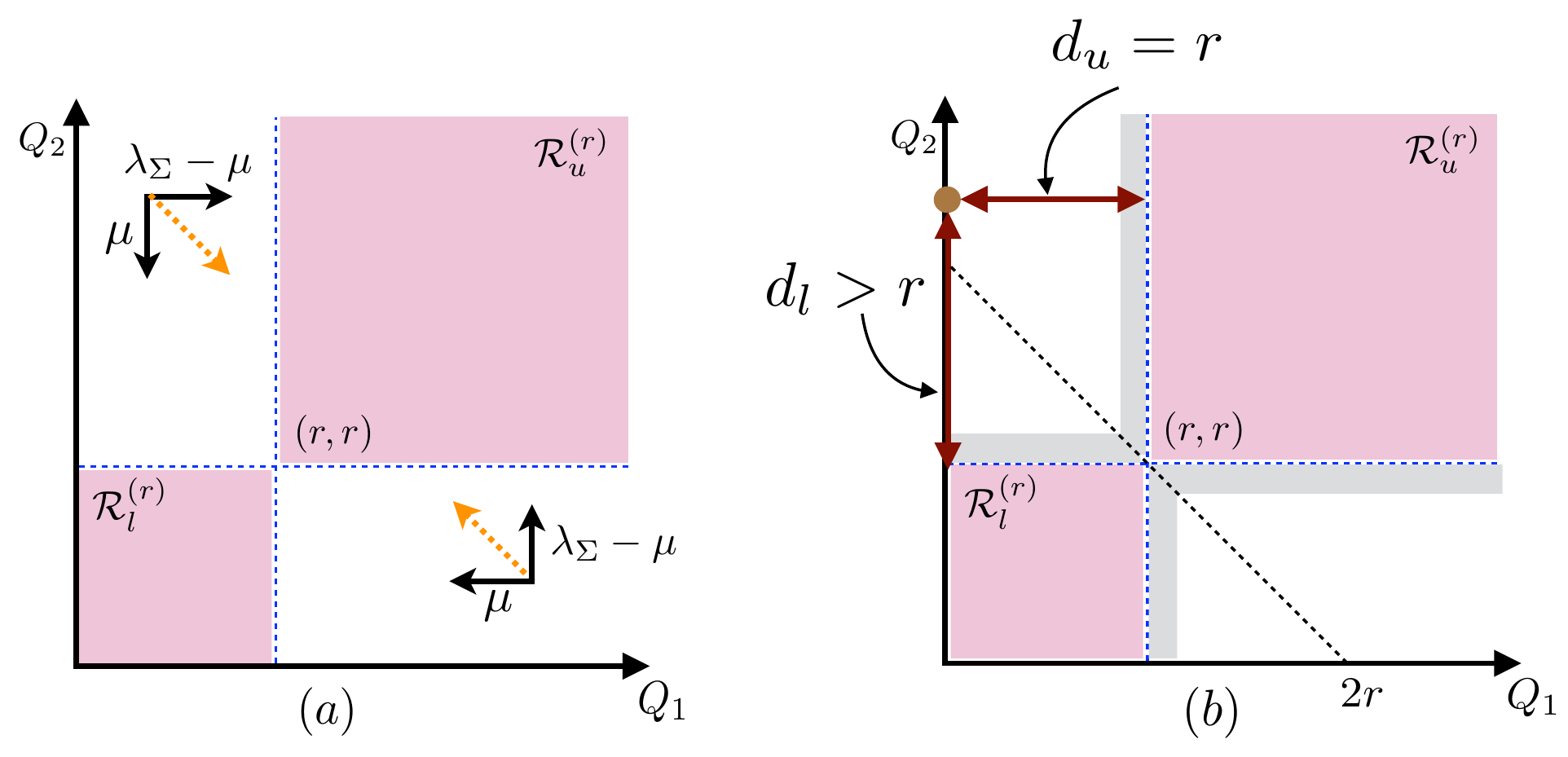}
	\caption{Geometric illustrations of the sufficient condition.}\label{fig:sufficient}
	\vspace{-4mm}
\end{figure}

Now, we turn to provide the high-level intuitions on Proposition \ref{thm:collapse} and Theorem \ref{thm:log} with the help of Fig. \ref{fig:sufficient}. This will facilitate the understanding of the results as well as their proofs.

To start with, note that by virtue of the JBT policy, when the queue-length state $\Q$ is outside the region $\mathcal{R}^{(r)}$, there always exists a positive drift towards the region $\mathcal{R}^{(r)}$. This is because in this case there exists a positive drift towards the lower region $\mathcal{R}_{l}^{(r)}$ and a positive drift towards the upper region $\mathcal{R}_{u}^{(r)}$, respectively (see Fig. 2(a) for an illustration).  
This provides the key intuition as to why the system state would concentrate around the region $\mathcal{R}^{(r)}$ since suppose there is no drift (e.g., under random routing) the expected distance to the region $\mathcal{R}^{(r)}$ would go to infinity as $r^{(\epsilon)}$ goes to infinity (assuming that the growth rate of $r^{(\epsilon)}$ is not too fast). In contrast, under the JBT policy, the distance remains constant (as shown by the gray color in Fig. \ref{fig:sufficient}(b)). This is the reason why we call it a state-space collapse result, which is different from much of previous works where the system state collapses to a lower dimensional space (e.g., a line or a convex cone) while our state-space collapse region $\mathcal{R}^{(r)}$ is of the same dimension as the original queue-length state vector. Hence, we need to develop new methods to apply this new type of state-space collapse result to achieve heavy-traffic delay optimality of the JBT policy, as in Theorem \ref{thm:log}. 

To this end, we will utilize the sufficient (and necessary) condition in Eq. \eqref{eq:necessary_main} again.
As discussed before, it basically requires us to guarantee that no server is idling while other servers are busy under high loads. 
To achieve this, a logarithmic growth rate as in Theorem \ref{thm:log} is sufficient. For an illustration of the main ideas behind the proof, let us consider a simple two-server case. In this case, Eq. \eqref{eq:necessary_main} reduces to 
\begin{align}
\label{eq:two_server}
	\lim_{\epsilon \downarrow 0}\ex{\overline{Q}_1^{(\epsilon)}(t+1) \overline{U}_2^{(\epsilon)} + \overline{Q}_2^{(\epsilon)}(t+1) \overline{U}_1^{(\epsilon)} } = 0.
\end{align}
Take the second term above for example, it can be rewritten as the summation of the following terms (for simplicity we omit the superscript $^{(\epsilon)}$)
\begin{align}
	&\overline{Q}_2(t+1) \overline{U}_1\mathcal{I}\left( \overline{Q}_2(t+1) \le 2r, \overline{Q}_1(t+1) = 0\right)\label{eq:first_term} \\
	&\overline{Q}_2(t+1) \overline{U}_1\mathcal{I}\left( \overline{Q}_2(t+1) > 2r, \overline{Q}_1(t+1) = 0\right)\label{eq:second_term},
\end{align}
where we use the fact that $Q_n(t+1)U_n(t) = 0$ again.
The expectation of Eq. \eqref{eq:first_term} can be upper bounded by $2 r^{(\epsilon)}\epsilon$ since $\ex{\overline{U}_1} \le \epsilon$. For the expectation of Eq. \eqref{eq:second_term}, we first apply Cauchy-Schwartz inequality and hence obtain its upper bound as
\begin{align*}
	C \frac{1}{\epsilon^2}\mathbb{P}\left(\overline{Q}_2(t+1) > 2r, \overline{Q}_1(t+1) = 0\right),
\end{align*}
where $C$ is a constant independent of $\epsilon$. Now, we can apply the state-space collapse result (i.e., Eq. \eqref{eq:collpase}) combined with Chernoff bound to show that the probability that one queue is empty and another queue length is larger than $2r$ has an exponential decay rate. In particular, we have 
\begin{align*}
	\mathbb{P}\left(\overline{Q}_2(t+1) > 2r, \overline{Q}_1(t+1) = 0\right) \lep{a} \mathbb{P}\left(d_{\mathcal{R}^{(r)}}\big(\overline{\Q}^{(\epsilon)}\big) \ge r \right) \lep{b} \frac{C^*}{e^{\theta^* r}},
\end{align*}
where (a) holds since in this case the distance to the region $\mathcal{R}^{(r)}$ is $r$ (see Fig. 2(b) for an illustration); (b) follows directly from state-space collapse result and Chernoff bound.
Therefore, combining the expectations of Eqs. \eqref{eq:first_term} and \eqref{eq:second_term}, yields
\begin{align*}
	\ex{\overline{Q}_2^{(\epsilon)}(t+1) \overline{U}_1^{(\epsilon)}} \le 2 r^{(\epsilon)}\epsilon + C_1\frac{1}{\epsilon^2}\frac{1}{e^{\theta^* r^{(\epsilon)}}},
\end{align*}
which approaches zero whenever $r^{(\epsilon)} = o(\frac{1}{\epsilon})$ and $r^{(\epsilon)} \ge K \log(1/\epsilon)$ where $K = 2(1+\alpha)/\theta^*$ for any $\alpha>0$. By the same arguments, we can establish the same result for the expectation of the first term in Eq. \eqref{eq:two_server}. Therefore, we have reached the sufficient condition for heavy-traffic delay optimality in Theorem \ref{thm:log}.


\section{Generalizations}
\label{sec:general}
For the illustration of the key ideas, the main results in the last section are obtained under the assumptions that both arrival and service processes have finite support. However, it is worth pointing out that the same results still hold (with only a change in constants) when the support is infinite. More specifically, we need the following weak condition on arrival and service processes, which requires that the tails of both arrival and service processes have an exponential decay.

\begin{condition}[Weaker condition on arrival and service]
	The i.i.d arrival process $A_{\Sigma}(t)$ and service process $S_{n}(t)$ satisfy 
	\begin{align*}
		\ex{e^{\theta_1 A_{\Sigma}(t)} } \le D_1 \text{ and } \ex{e^{\theta_2 S_{n}(t)} } \le D_2,
	\end{align*}
	for each $n$ where the constants $\theta_1 > 0$, $\theta_2 > 0$, $D_1 <\infty$ and $D_2 <\infty$ are all independent of $\epsilon$.
\end{condition}



In order to obtain the same main results under the weaker condition above, we should make some mild changes in our proofs. In the following, we will highlight the key steps involved in this process.

(i) First, note that in order to establish condition (C1) in Lemma \ref{lem:basis}, we would use the following upper bound in our proofs based on the finite support assumptions.
\begin{align*}
	\ex {\norm{\A(t_0) - \s(t_0)}^2  \mid Z(t_0) } \le L \triangleq N \max(A_{max}, S_{max})^2.
\end{align*}
However, under the weaker Condition A, we can still bound the left-hand side by a constant independent of $\epsilon$. This directly follows from the fact that all the moments of a random variable are finite if its moment generating function is finite in an open interval containing zero. 

(ii) Second, we should now replace condition (C2) in Lemma \ref{lem:basis} with the following weak stochastic domination condition (C2$^\prime$),
\begin{itemize}
	\item (C2$^\prime$) $\left[\Delta V(X) \mid X(t_0) = X\right] \prec W$ for all $t_0$ and $\ex{e^{\theta W}} = D$ is finite for some $\theta > 0$.
\end{itemize}
This condition holds under the weaker Condition A since the arrival and service processes both have an exponentially bounded tail by the finiteness of their moment generating functions. As shown by Theorem 2.3 in~\cite{hajek1982hitting}, the combination of (C1) and (C2$^\prime$) is sufficient to guarantee bounded moments as required in the proof of our main results.

(iii) Third, we now should take a careful treatment of the unused service. For example, the following result plays a key role in establishing the necessary and sufficient condition in Eq. \eqref{eq:necessary_main}
\begin{align*}
	\lim_{\epsilon \downarrow 0}\ex{\norms{\overline{\UU}^{(\epsilon)}}^2_1} = 0.
\end{align*}
Under the assumption of finite support for the service process, the left-hand side can be easily bounded above by $NS_{max}\epsilon$, which approaches zero as $\epsilon \to 0$. Now, under the weak condition, we need to adopt the truncation trick to handle the unbounded service. More specifically, let us consider any $n \in \mathcal{N}$, we have for any $t \ge 0$ and constant $S^{\prime}$
\begin{align*}
	{U}_n^2(t) &\le U_n(t)S_n(t)\\
	& = U_n(t)S_n(t) \mathcal{I}\left(S_n(t) \le S^{\prime} \right) + U_n(t)S_n(t) \mathcal{I}\left(S_n(t) > S^{\prime} \right)\\
	& \le U_n(t) S^{\prime} + S_n^2(t)\mathcal{I}\left(S_n(t) > S^{\prime} \right).
\end{align*}
In steady state, we have 
\begin{align*}
	\ex{\overline{U}_n^2} &\le \ex{\overline{U}_n}S^{\prime} + \ex{S_n^2(\infty)\mathcal{I}\left(S_n(\infty) > S^{\prime} \right)}\\
	&\lep{a} \epsilon S^{\prime} + \ex{S_n^2(0)\mathcal{I}\left(S_n(0) > S^{\prime} \right)}\\
	&\lep{b} \epsilon S^{\prime} + \beta,
\end{align*}
where (a) follows from the fact that $\ex{\norms{\overline{\UU}^{(\epsilon)}}_1} = \epsilon$ and service process is \emph{i.i.d.}; in (b), we choose $S^{\prime}$ such that $\ex{S_n^2(0)\mathcal{I}\left(S_n(0) > S^{\prime} \right)} \le \beta$, which is possible by the exponential decay rate of $S_n(0)$ under the weak condition. Thus, we have 
\begin{align*}
	\lim_{\epsilon \downarrow 0} \ex{\overline{U}_n^2} \le \beta,
\end{align*}
for any $\beta > 0$. Hence, we have $\lim_{\epsilon \downarrow 0} \ex{\overline{U}_n^2} = 0$ for each $n$.

\begin{remark}
The three highlighted key steps could also demonstrate their generalization power in previous works where the Lyapunov drift-based framework is adopted under the assumption of finite supports for the arrival and service processes.
\end{remark}

\section{Proofs}
In this paper, we will adopt the Lyapunov drift-based approach developed in~\cite{eryilmaz2012asymptotically} to derive bounded moments in steady state. In particular, the following lemma, which follows directly from Lemmas 2 and 3 in~\cite{maguluri2016heavy}, will be the main tool in our proofs.
\begin{lemma}
	      \label{lem:basis}
	        For an irreducible aperiodic and positive recurrent Markov chain $\{X(t), t \ge 0\}$ over a countable state space $\mathcal{X}$, which converges in distribution to $\overline{X}$,  and suppose $V: \mathcal{X} \rightarrow \mathbb{R}_{+}$ is a Lyapunov function. We define the drift of $V$ at $X$ as 
	        \[\Delta V(X)\triangleq [V(X(t_0+1)) - V(X(t_0))] \mathcal{I}(X(t_0) = X),\]
	        where $\mathcal{I}(.)$ is the indicator function. Suppose the drift of $V$ satisfies the following conditions:

	        \begin{itemize}
	          \item (C1) There exists an $\eta> 0$ and a $\kappa <  \infty$ such that for any $t_0 = 1,2,\ldots$ and for all $X \in \mathcal{X}$ with $V(X)\ge \kappa$, 
	          \[\mathbb{E}\left[\Delta V(X) \mid X(t_0) = X\right]\le -\eta.\]
	          \item (C2) There exists a constant $D < \infty$ such that for all $X\in \mathcal{X}$,
	          \[\mathbb{P}(|\Delta V(X)| \le D) = 1.\]
	        \end{itemize}

	        Then $\{V(X(t)), t\ge0\}$ converges in distribution to a random variable $\overline{V}$ for which there exists a $\theta^* >0$ and a $C^* < \infty$ such that 
	        \begin{align*}
	        	\ex{e^{\theta^* \overline{V}}} \le C^*,
	        \end{align*}
	        which directly implies that all the moments of $\overline{V}$ exist and are finite.
	        More specifically, we have for any $p = 1,2,\ldots$
	        \begin{equation}
	        \label{eq:upper_siva}
	        	\ex{V(\overline{X})^p} \le (2\kappa)^p + (4D)^p\left(\frac{D+\eta}{\eta} \right)^p p!.
	        \end{equation}
\end{lemma}

We would also utilize the following useful result in our proofs.
	\begin{lemma}
	\label{lem:necessary}
		For the JBT policy with threshold $r \ge 1$, it is heavy-traffic delay optimal if and only if 
		\begin{align}
		\label{eq:necessary}
			\lim_{\epsilon \downarrow 0}\ex{\big\lVert\overline{\Q}^{(\epsilon)}(t+1) \big\rVert_1 \big\lVert\overline{\UU}^{(\epsilon)}(t) \big\rVert_1} = 0.
		\end{align}
	\end{lemma}
	This lemma is a direct application of the results in~\cite{zhou2018flexible}, which establishes that Eq. \eqref{eq:necessary} is the sufficient and necessary condition for any load balancing policy to be heavy-traffic delay optimal if the system is stable with bounded moments. 
	By Lemma \ref{thm:throughput}, we have that the JBT policy is throughput optimal with all the moments being bounded for any $r\ge 1$, and hence the above lemma holds. 

\subsection{Proof of Theorem \ref{thm:constant}}
\label{sec:proof_thm_constant}
Before we present our proof, we first give the following useful result, which can be established by setting the mean drift a chosen Lyapunov function to zero in steady state. For completeness, the proof is given at Appendix \ref{sec:proof_lemma5}.
\begin{lemma}
\label{lem:key_equation}
Consider a load balancing system with homogeneous servers under the JBT policy. For any threshold $r \ge 1$, we have 
\begin{align*}
	2\sum_{i=1}^N\sum_{j>i}^N \ex{\left((\overline{Q}_i^+)^{(\epsilon)}\overline{U}_j^{(\epsilon)} + (\overline{Q}_j^+)^{(\epsilon)}U_i^{(\epsilon)}\right)} = \mathcal{T}_1^{(\epsilon)} + \mathcal{T}_2^{(\epsilon)} - \mathcal{T}_3^{(\epsilon)},
\end{align*}	
where 
\begin{align*}
	&\mathcal{T}_1^{(\epsilon)} \triangleq 2\sum_{i=1}^N\sum_{j>i}^N \ex{\left(\overline{Q}_i^{(\epsilon)} - \overline{Q}_j^{(\epsilon)}\right)\left(\overline{A}_i^{(\epsilon)} - \overline{A}_j^{(\epsilon)}\right)}\\
	&\mathcal{T}_2^{(\epsilon)} \triangleq \sum_{i=1}^N\sum_{j>i}^N\ex{\left( \overline{A}_i^{(\epsilon)} - \overline{A}_j^{(\epsilon)} - \overline{S}_i^{(\epsilon)} +\overline{S}_j^{(\epsilon)}\right)^2} \\
	&\mathcal{T}_3^{(\epsilon)} \triangleq \sum_{i=1}^N\sum_{j>i}^N\ex{\left(\overline{U}_i^{(\epsilon)} - \overline{U}_j^{(\epsilon)}\right)^2}\\
	&\overline{\Q}^+ \triangleq  \overline{\Q}(t+1)
\end{align*}
and $\overline{A}_i^{(\epsilon)}$ and $\overline{U}_i^{(\epsilon)}$ are dependent of $\overline{\Q}$ for each $i$ and $\epsilon > 0$.
\end{lemma}


Now, we are ready to present the proof of Theorem \ref{thm:constant}.
\begin{proof}[Proof of Theorem \ref{thm:constant}]
	To start with, we first note that the sufficient and necessary condition in Lemma \ref{lem:necessary} can be rewritten as follows under the JBT policy.
	\begin{align}
		&2 \ex{\big\lVert\overline{\Q}^{(\epsilon)}(t+1) \big\rVert_1 \big\lVert\overline{\UU}^{(\epsilon)}(t) \big\rVert_1}\nonumber\\
		\ep{a} &2\sum_{i=1}^N\sum_{j>i}^N \ex{\left((\overline{Q}_i^+)^{(\epsilon)}\overline{U}_j^{(\epsilon)} + (\overline{Q}_j^+)^{(\epsilon)}U_i^{(\epsilon)}\right)}\nonumber\\
		\ep{b} &4\sum_{i=1}^N\sum_{j>i}^N\ex{\left((\overline{Q}_i^+)^{(\epsilon)}\overline{U}_j^{(\epsilon)} \right)}\nonumber\\
		\ep{c} &4\sum_{i=1}^N\sum_{j>i}^N\left(\sum_{k=1}^\infty k \overline{U}_j^{(\epsilon)} \mathbb{P} \left( (\overline{Q}_i^+)^{(\epsilon)}= k, (\overline{Q}_i^+)^{(\epsilon)}= 0, \overline{U}_j^{(\epsilon)} \ge 1 \right)\right), \label{eq:T0}
	\end{align}
	in which (a) and (c) follow from the fact $Q_i(t+1)U_i(t) = 0$ for each $i$ and $t\ge 0$; (b) holds by the symmetry property of JBT policy for homogeneous servers.

	 Thus, by Lemma \ref{lem:necessary}, Lemma \ref{lem:key_equation} and the above equation, in order to analyze heavy-traffic delay optimality of JBT under any constant threshold, all we need to do is to focus on terms $\mathcal{T}_1^{(\epsilon)}$, $\mathcal{T}_2^{(\epsilon)}$ and $\mathcal{T}_3^{(\epsilon)}$, respectively. 

	 Now, let us first focus on case (1) in Theorem \ref{thm:constant}.

	 For $\mathcal{T}_1^{(\epsilon)}$, we have 
	 \begin{align}
	 \label{eq:T1_first}
		\mathcal{T}_1^{(\epsilon)} &\triangleq 2\sum_{i=1}^N\sum_{j>i}^N \ex{\left(\overline{Q}_i^{(\epsilon)} - \overline{Q}_j^{(\epsilon)}\right)\left(\overline{A}_i^{(\epsilon)} - \overline{A}_j^{(\epsilon)}\right)}\nonumber\\
		& = 2\sum_{i=1}^N\sum_{j>i}^N \ex{\left(\overline{Q}_i- \overline{Q}_j\right)\left(\overline{A}_i - \overline{A}_j\right)\mathcal{I}\left(\overline{Q}_i \ge r, \overline{Q}_j \ge r\right) }\nonumber\\
		& \quad + 2\sum_{i=1}^N\sum_{j>i}^N \ex{\left(\overline{Q}_i- \overline{Q}_j\right)\left(\overline{A}_i - \overline{A}_j\right)\mathcal{I}\left(\overline{Q}_i < r, \overline{Q}_j < r\right) }\nonumber\\
		& \quad + 4 \sum_{i=1}^N\sum_{j>i}^N \ex{\left(\overline{Q}_i- \overline{Q}_j\right)\left(\overline{A}_i - \overline{A}_j\right)\mathcal{I}\left(\overline{Q}_i \ge r, \overline{Q}_j < r\right) }\nonumber\\
		& \ep{a} 4 \sum_{i=1}^N\sum_{j>i}^N \ex{\left(\overline{Q}_i- \overline{Q}_j\right)\left(\overline{A}_i - \overline{A}_j\right)\mathcal{I}\left(\overline{Q}_i \ge r, \overline{Q}_j < r\right) }\nonumber\\
		& \gep{b} -4\lambda_\Sigma \sum_{i=1}^N\sum_{j>i}^N \sum_{m=0}^{r-1} \sum_{k=r}^\infty (k-m)  \mathbb{P}\left(\overline{Q}_i = k, \overline{Q}_j = m \right)\nonumber\\
		& \ep{c} -4\lambda_\Sigma \sum_{i=1}^N\sum_{j>i}^N \sum_{m=0}^{r-1} \sum_{k=r}^\infty (k-m)  \mathbb{P}\left(\overline{Q}_i^+ = k, \overline{Q}_j^+ = m \right),
	\end{align}
	where (a) follows from the definition of the JBT policy, i.e., when both queues are in memory or both queues are not in memory, they have the same probability to be selected in the homogeneous case; (b) is true since when the ID of server $j$ is in $m(t)$ while the ID of server $i$ is not, we have $A_i(t) = 0$ and $A_j(t)\le A_{\Sigma}(t)$ by the definition of the JBT policy; (c) holds since $\overline{\Q}(t+1)$ has the same distribution as $\overline{\Q}(t)$ in steady state.
	
	In order to further simplify the term $\mathcal{T}_1^{(\epsilon)}$, we need to define the following events in which $k\ge r$ and $1 \le m \le r-1$.
	\begin{align*}
		&E_{(k,m)} \triangleq \left\{\overline{Q}_i^+ = k, \overline{Q}_j^+ = m  \right\}\\
		&E_{(k,m)}^+ \triangleq \left\{\overline{Q}_i(t+2) = k, \overline{Q}_j(t+2) = m  \right\}\\
		&E_{(k,0,0)} \triangleq \left\{\overline{Q}_i^+ = k, \overline{Q}_j^+ = 0, \overline{U}_j = 0  \right\}\\
		&E_{(k,0,0)}^+ \triangleq \left\{\overline{Q}_i(t+2) = k, \overline{Q}_j(t+2) = 0, \overline{U}_j^+ = 0  \right\}\\
		&E_{(k,0,\ge1)} \triangleq \left\{\overline{Q}_i^+ = k, \overline{Q}_j^+ = 0, \overline{U}_j \ge 1 \right\}\\
		&E_{(k,0,\ge 1)}^+ \triangleq \left\{\overline{Q}_i(t+2) = k, \overline{Q}_j(t+2) = 0, \overline{U}_j^+ \ge 1 \right\}.
	\end{align*}
	Note that by the assumptions of arrival and service processes, there exists a positive probability $\hat{p}$ (independent of $\epsilon$) such that there is no arrival and meanwhile the potential service of all the servers are $d$ for some $d$ between $1$ and $S_{max}$. For ease of exposition, we take $d = 1$ in the following proof, and the same techniques apply for the case where $d \neq 1$. Now, for each occurrence of event $E_{(k,m)}$, there exists a positive probability $\hat{p}$ such that $E_{(k-1,m-1)}^+$ will happen. Therefore, we have
	\begin{align}
	\label{eq:1}
		\mathbb{P}\left(E_{(k-1,m-1)}\right) \ep{a} \mathbb{P}\left(E_{(k-1,m-1)}^+\right) \ge \hat{p} \mathbb{P}\left(E_{(k,m)}\right),
	\end{align}
	where (a) holds due to the fact that both events are defined in steady state. Similarly, we have 
	\begin{align}
		&\mathbb{P}\left(E_{(k-1,0,0)}\right) = \mathbb{P}\left(E_{(k-1,0,0)}^+\right) \ge \hat{p} \mathbb{P}\left(E_{(k,1)}\right)\label{eq:2}\\
		&\mathbb{P}\left(E_{(k-1,0,\ge 1)}\right)  = \mathbb{P}\left(E_{(k-1,0,\ge1)}^+\right) \ge \hat{p}  \mathbb{P}\left(E_{(k,0,0)}\right)\label{eq:3}.
	\end{align}
	Now, we can further simplify $\mathcal{T}_1^{(\epsilon)}$ as follows
	\begin{align}
	\label{eq:T1}
		\mathcal{T}_1^{(\epsilon)} &\gep{a} -4\lambda_\Sigma \sum_{i=1}^N\sum_{j>i}^N \left(\sum_{k=r}^\infty k \mathbb{P}\left(E_{(k,0,\ge1)}\right) + \frac{1}{\hat{p}}\sum_{k=r}^\infty k \mathbb{P}\left(E_{(k-1,0,\ge1)}\right) \right)\nonumber \\
		& \quad -4\lambda_\Sigma \sum_{i=1}^N\sum_{j>i}^N \left(\sum_{m=1}^{r-1} \sum_{k=r}^\infty \frac{1}{\hat{p}^{m+1}} (k-m)\mathbb{P}\left(E_{(k-m-1,0,\ge1)}\right) \right)\nonumber\\
		& = -4\lambda_\Sigma \sum_{i=1}^N\sum_{j>i}^N \left( \sum_{l = 0}^r \sum_{h=r-l}^\infty \frac{1}{\hat{p}^l}h \mathbb{P}\left(E_{(h,0,\ge1)}\right) \right)\nonumber\\
		& \quad -4\lambda_\Sigma \sum_{i=1}^N\sum_{j>i}^N\left( \sum_{l = 1}^r \sum_{h=r-l}^\infty \frac{1}{\hat{p}^l} \mathbb{P}\left(E_{(h,0,\ge1)}\right)\right)\nonumber\\
		& \gep{b} -4\lambda_\Sigma \sum_{i=1}^N\sum_{j>i}^N \left( \sum_{l = 0}^r \frac{1}{\hat{p}^l} \sum_{h=0}^\infty h\overline{U}_j \mathbb{P}\left(E_{(h,0,\ge1)}\right) \right)\nonumber\\
		& \quad -4\lambda_\Sigma \sum_{i=1}^N\sum_{j>i}^N\left( \sum_{l = 1}^r \frac{1}{\hat{p}^l} \epsilon \right),
	\end{align}
	where (a) follows from Eqs. \eqref{eq:1}, \eqref{eq:2} and \eqref{eq:3}; (b) holds since $U_j(t) \ge 1$ and $\ex{\overline{U}_j} \le \ex{\norms{\overline{\UU}^{(\epsilon)}}_1} = \epsilon$. The latter fact can be easily obtained by setting mean drift of $\hat{V}(Z(t)) \triangleq \norm{\Q(t)}_1$ to be zero in steady state, which is true since all the moments of $\norms{\overline{\Q}}$ is bounded.

	For $\mathcal{T}_2^{(\epsilon)}$, we can simplify it as follows.
	\begin{align}
\label{eq:T2}
	\mathcal{T}_2^{(\epsilon)} &\triangleq \sum_{i=1}^N\sum_{j>i}^N\ex{\left( \overline{A}_i^{(\epsilon)} - \overline{A}_j^{(\epsilon)} - \overline{S}_i^{(\epsilon)} +\overline{S}_j^{(\epsilon)}\right)^2}\nonumber \\
	& \ep{a}\sum_{i=1}^N\sum_{j>i}^N\ex{\left( \overline{A}_i^{(\epsilon)} - \overline{A}_j^{(\epsilon)}\right)^2 - \left(\overline{S}_i^{(\epsilon)} -\overline{S}_j^{(\epsilon)}\right)^2}\nonumber\\
	& \ep{b} \left(N-1\right)\left(\left(\sigma_{\Sigma}^{(\epsilon)}\right)^2+\left(\lambda_{\Sigma}^{(\epsilon)}\right)^2 + \nu_{\Sigma}^2\right),
\end{align}
where (a) holds since the arrival and service are independent and the servers are homogeneous; (b) is true because $A_i(t)A_j(t) = 0$ for all $i\neq j$ and $t\ge0$, and the service is independent and homogeneous.

	For $\mathcal{T}_3^{(\epsilon)}$, we can simplify it as follows.
	\begin{align}
	\label{eq:T3}
		\mathcal{T}_3^{(\epsilon)} &\triangleq \sum_{i=1}^N\sum_{j>i}^N\ex{\left(\overline{U}_i^{(\epsilon)} - \overline{U}_j^{(\epsilon)}\right)^2}\nonumber\\
		 & \lep{a} \left(N-1\right) \ex{\norms{\overline{\UU}^{(\epsilon)}}_1^2 }\nonumber\\
		 & \lep{b} \epsilon\left(N-1\right) S_{max},
	\end{align}
	where (a) follows from the fact that $U_n(t) \ge 0$ for any $n \in \mathcal{N}$; (b) holds because of $U_n(t) \le S_{max}$ for any $n\in \mathcal{N}$ and the fact $\ex{\norms{\overline{\UU}^{(\epsilon)}}_1} = \epsilon$.
	
	Now, substituting Eqs. \eqref{eq:T0}, \eqref{eq:T1}, \eqref{eq:T2} and \eqref{eq:T3} into the equation in Lemma \ref{lem:key_equation}, yields
	\begin{align*}
	&4\sum_{i=1}^N\sum_{j>i}^N\left(\sum_{k=1}^\infty k \overline{U}_j^{(\epsilon)} \mathbb{P} \left( (\overline{Q}_i^+)^{(\epsilon)}= k, (\overline{Q}_i^+)^{(\epsilon)}= 0, \overline{U}_j^{(\epsilon)} \ge 1 \right)\right)\\
	 = &4\sum_{i=1}^N\sum_{j>i}^N\left(\sum_{k=1}^\infty k\overline{U}_j\mathbb{P}\left(E_{(k,0,\ge1)}\right) \right)\\
	\ge &-4\lambda_\Sigma \sum_{i=1}^N\sum_{j>i}^N \left( \sum_{l = 0}^r \frac{1}{\hat{p}^l} \sum_{h=0}^\infty h \overline{U}_j\mathbb{P}\left(E_{(h,0,\ge1)}\right) \right)-4\lambda_\Sigma \sum_{i=1}^N\sum_{j>i}^N\left( \sum_{l = 1}^r \frac{1}{\hat{p}^l} \epsilon \right) \\
		 & + \left(N-1\right)\left( \left(\sigma_\Sigma^{(\epsilon)}\right)^2 + \left(\lambda_\Sigma^{(\epsilon)}\right)^2 + \nu_{\Sigma}^2\right) - S_{max}\left(N-1\right) \epsilon,
	\end{align*}
	which can be simplified as 
	\begin{align*}
	&\left(4+ 4\lambda_\Sigma \sum_{l = 0}^r \frac{1}{\hat{p}^l}\right)	\sum_{i=1}^N\sum_{j>i}^N\left(\sum_{k=1}^\infty k\overline{U}_j\mathbb{P}\left(E_{(k,0,\ge1)}\right) \right) \ge  - S_{max}\left(N-1\right) \epsilon\\
	& \quad -4\lambda_\Sigma \sum_{i=1}^N\sum_{j>i}^N\left( \sum_{l = 1}^r \frac{1}{\hat{p}^l} \epsilon \right) + \left(N-1\right)\left( \left(\sigma_\Sigma^{(\epsilon)} \right)^2 + \left(\lambda_\Sigma^{(\epsilon)}\right)^2 + \nu_{\Sigma}^2\right).
	\end{align*}
	Then taking liminf on both sides gives 
	\begin{align}
	\label{eq:liminfge0}
		\liminf_{\epsilon \downarrow 0} \sum_{i=1}^N\sum_{j>i}^N\left(\sum_{k=1}^\infty k\overline{U}_j\mathbb{P}\left(E_{(k,0,\ge1)}\right) \right) \ge \frac{(N-1)\left(\sigma_{\Sigma}^2 + \mu_{\Sigma}^2 + \nu_{\Sigma}^2\right)}{4+ 4\mu_\Sigma \sum_{l = 0}^r \frac{1}{\hat{p}^l}}>0
	\end{align}
	which holds since threshold $r$ is a constant and $\hat{p}$ would not vanish as $\epsilon \to 0$. Therefore, by Lemma \ref{lem:necessary} and Eq. \eqref{eq:T0}, we have 
	\begin{align*}
		\liminf_{\epsilon \downarrow 0} \epsilon \ex{\sum_{n=1}^N \overline{Q}_n^{(\epsilon)} }  > \frac{\zeta}{2},
	\end{align*}
	where $\zeta$ is the constant defined as in Lemma \ref{lem:lower_bound}. 

	To establish the inequality \eqref{eq:limsup} in Theorem \ref{thm:constant}, note that the term $\mathcal{T}_1^{(\epsilon)}$ is equal to $0$ for any $\epsilon > 0$ under random routing, and $\mathcal{T}_2^{(\epsilon)}$ and $\mathcal{T}_3^{(\epsilon)}$ converge to the same constant for both random routing and JBT. Thus, based on Lemma \ref{lem:necessary} and Lemma \ref{lem:key_equation}, all we need to show is that under the JBT policy $\limsup_{\epsilon \downarrow 0}\mathcal{T}_1^{(\epsilon)} < 0$. To this end, we can upper bound it as follows by reusing the equation (a) in Eq. \eqref{eq:T1_first}.
	\begin{align*}
		\mathcal{T}_1^{(\epsilon)} &= 4 \sum_{i=1}^N\sum_{j>i}^N \ex{\left(\overline{Q}_i- \overline{Q}_j\right)\left(\overline{A}_i - \overline{A}_j\right)\mathcal{I}\left(\overline{Q}_i \ge r, \overline{Q}_j < r\right) }\nonumber\\
		&\lep{a} -\frac{4\lambda_\Sigma}{N-1} \sum_{i=1}^N\sum_{j>i}^N \sum_{m=0}^{r-1} \sum_{k=r}^\infty (k-m)  \mathbb{P}\left(\overline{Q}_i = k, \overline{Q}_j = m \right)\nonumber\\
		& \le-\frac{4\lambda_\Sigma}{S_{max}(N-1)} \sum_{i=1}^N\sum_{j>i}^N \left(\sum_{k=1}^\infty k\overline{U}_j\mathbb{P}\left(E_{(k,0,\ge1)}\right) \right),
	\end{align*}
	where (a) holds since when $Q_i(t)\ge r$ and $Q_j(t) < r$, the lower bound on the probability of server $j$ being chosen under JBT is $1/(N-1)$.
	Now, taking limsup on both sides, yields
	\begin{align*}
		\limsup_{\epsilon \downarrow 0}\mathcal{T}_1^{(\epsilon)} &\le - \frac{4\lambda_\Sigma}{S_{max}(N-1)} \liminf_{\epsilon \downarrow 0} \sum_{i=1}^N\sum_{j>i}^N \left(\sum_{k=1}^\infty k\overline{U}_j\mathbb{P}\left(E_{(k,0,\ge1)}\right) \right)\\
		&<0,
	\end{align*}
	where the last inequality follows directly from Eq. \eqref{eq:liminfge0}. Hence, we have completed the proof of the first case in Theorem \ref{thm:constant}.

	Now, let us turn to case (2) in Theorem \ref{thm:constant}. Based on the discussions above, in order to show that the JBT policy with $r^{(\epsilon)} = (1/\epsilon)^{1+\alpha}$ and $\alpha > 0$ achieves the same limit as random routing, all we need to show is that $\lim_{\epsilon \downarrow 0} \mathcal{T}_1^{(\epsilon)} = 0$.
	Again, using the equation (a) in Eq. \eqref{eq:T1_first}, we obtain 
	\begin{align*}
		\mathcal{T}_1^{(\epsilon)} &= 4 \sum_{i=1}^N\sum_{j>i}^N \ex{\left(\overline{Q}_i- \overline{Q}_j\right)\left(\overline{A}_i - \overline{A}_j\right)\mathcal{I}\left(\overline{Q}_i \ge r, \overline{Q}_j < r\right) }\nonumber\\
		& \ge -4\lambda_{\Sigma} \sum_{i=1}^N\sum_{j>i}^N \sum_{m=0}^{r-1}\ex{(\overline{Q}_i - m)\mathcal{I}\left(\overline{Q}_i \ge r, \overline{Q}_j =m\right)}\\
		& \ge -4\lambda_{\Sigma} \sum_{i=1}^N\sum_{j>i}^N \sum_{m=0}^{r-1}\ex{\overline{Q}_i\mathcal{I}\left(\overline{Q}_i \ge r, \overline{Q}_j =m\right)}\\
		& \ge -4\lambda_{\Sigma} \sum_{i=1}^N\sum_{j>i}^N \sum_{m=0}^{r-1}\sqrt{\ex{\overline{Q}_i^2}\mathbb{P}\left(\overline{Q}_i \ge r, \overline{Q}_j =m\right)}\\
		& \gep{a} -4\lambda_{\Sigma} \sum_{i=1}^N\sum_{j>i}^N \sum_{m=0}^{r-1}\sqrt{\frac{M^{\prime}}{\epsilon^2}\frac{e^{\theta^* (1/\epsilon)}}{e^{\theta^* r}}},
	\end{align*}
	where (a) follows from the bounded moments in Lemma \ref{thm:throughput} and Chernoff bound based on Eq. \eqref{eq:mgf} in the proof of Lemma~\ref{thm:throughput}. Thus, if $r^{(\epsilon)} = (1/\epsilon)^{1+\alpha}$ for any constant $\alpha > 0$, we have $\lim_{\epsilon \downarrow 0}\mathcal{T}_1^{(\epsilon)} = 0$. Hence, we have established the second case in Theorem \ref{thm:constant}.
\end{proof}

\subsection{Proof of Theorem \ref{thm:log}}
\label{sec:proof_theorem_log}
\begin{proof}[Proof of Theorem \ref{thm:log}]
	Based on the result in Lemma~\ref{lem:necessary}, in order to prove Theorem~\ref{thm:log}, we need just focus on the left-hand side of Eq.~\eqref{eq:necessary}. Let us first define 
	\begin{align*}
		\mathcal{T}^{(\epsilon)} &\triangleq \ex{\big\lVert\overline{\Q}^{(\epsilon)}(t+1) \big\rVert_1 \big\lVert\overline{\UU}^{(\epsilon)}(t) \big\rVert_1}\\
		& = \ex{\sum_{i=1}^N \overline{U}_i \left(\sum_{j=1}^N \overline{Q}_j^+ \right) },
	\end{align*}
	in which for brevity we omit the references $t$ and $\epsilon$, and use $\overline{\Q}^+$ to denote $\overline{\Q}(t+1)$. Thus, all we need to show is that $\lim_{\epsilon \downarrow 0} \mathcal{T}^{(\epsilon)} = 0$ under the assumptions of Theorem~\ref{thm:log}. Since $\overline{U}_i\overline{Q}_i^+ = 0$ by the queue-length dynamic in Eq. \eqref{eq:Qdynamic}, we have for each $i \in \mathcal{N}$,
	\begin{align}
	 &\ex{\overline{U}_i \left(\sum_{j=1}^N \overline{Q}_j^+ \right) }\nonumber\\
	 = &\ex{\overline{U}_i \left(\sum_{j=1}^N \overline{Q}_j^+ \right) \mathcal{I}{ \left(\overline{Q}_i^+ = 0\right) } }\nonumber\\
	  = &\ex{\overline{U}_i \left(\sum_{j=1}^N \overline{Q}_j^+ \right) \mathcal{I}{ \left(\overline{Q}_i^+ = 0, \max_j \overline{Q}_j^+ \le r\sqrt{N-1} + r \right) } }\label{eq:op_lower}\\
	  &+ \ex{\overline{U}_i \left(\sum_{j=1}^N \overline{Q}_j^+ \right) \mathcal{I}{ \left(\overline{Q}_i^+ = 0, \max_j \overline{Q}_j^+ > r\sqrt{N-1} + r \right) } }\label{eq:op_upper}.
	\end{align}
	Now, it remains to show that both Eqs. \eqref{eq:op_lower} and \eqref{eq:op_upper} approach $0$ as $\epsilon \to 0$.
	To start with, we can bound Eq. \eqref{eq:op_lower} as follows.
	\begin{align*}
		&\ex{\overline{U}_i \left(\sum_{j=1}^N \overline{Q}_j^+ \right) \mathcal{I}{ \left(\overline{Q}_i^+ = 0, \max_j \overline{Q}_j^+ \le r\sqrt{N-1} + r\right) } }\\
		\le & r(N-1)(\sqrt{N-1}+1)\ex{\overline{U}_i }\\
		\le & r(N-1)(\sqrt{N-1}+1)\epsilon,
	\end{align*}
	where the last inequality follows from the fact $\ex{\norms{\overline{\UU}^{(\epsilon)}}_1} = \epsilon$. Thus, Eq. \eqref{eq:op_lower} approaches $0$ as $\epsilon \to 0$ since $r^{(\epsilon)} = o(1/\epsilon)$. 

	Then, we can turn to bound Eq. \eqref{eq:op_upper} in the following way.
	\begin{align*}
		&\ex{\overline{U}_i \left(\sum_{j=1}^N \overline{Q}_j^+ \right) \mathcal{I}{ \left(\overline{Q}_i^+ = 0, \max_j \overline{Q}_j^+ > \sqrt{N-1}r + r \right) } }\\
		\lep{a} & S_{max}\ex{\norm{\overline{\Q}}_1 \mathcal{I}{ \left(\overline{Q}_i = 0, \max_j \overline{Q}_j > \sqrt{N-1}r + r \right) } }\\
		\lep{b} & S_{max}\sqrt{ \ex{\norm{\overline{\Q}}_1^2} \mathbb{P}\left( d_{\mathcal{R}^{(r)}}\left(\overline{\Q}\right) \ge r \right)}\\
		\lep{c} & S_{max} \sqrt{M_2 \frac{1}{\epsilon^2}\frac{C^*}{e^{\theta^*r}}},
	\end{align*}
	where (a) follows from the fact that $U_i(t)\le S_i(t) \le S_{max}$ for any $i \in \mathcal{N}$ and $t\ge0$; (b) holds due to Cauchy-Schwartz inequality and the following facts. For any system state $Z(t)$ that satisfies $Q_i(t) = 0$ for some $i$ and $\max_j \overline{Q}_j > \sqrt{N-1}r + r$, we have 
	\begin{align*}
		 d_{\mathcal{R}_l^{(r)}}(\Q(t)) > r\sqrt{N-1}\\
		r \le  d_{\mathcal{R}_u^{(r)}}(\Q(t)) \le r\sqrt{N-1}.
	\end{align*}
	Thus, 
	\begin{align*}
		d_{\mathcal{R}^{(r)}}(\Q(t)) = \min\{d_{\mathcal{R}_l^{(r)}}(\Q(t)),  d_{\mathcal{R}_u^{(r)}}(\Q(t))\} \ge r,
	\end{align*}
	and hence we have (b). The inequality (c) comes from the Chernoff bound, the moments bound in Lemma \ref{thm:throughput} and state-space collapse in Proposition \ref{thm:collapse}, in which the constants $M_2$, $C^*$ and $\theta^*$ are all independent of $\epsilon$. Now, under the condition that $r^{(\epsilon)} \ge K\log(1/\epsilon)$ where $K = 2(1+\alpha)/\theta^*$ and $\alpha>0$, we have that Eq. \eqref{eq:op_upper} approaches zero as $\epsilon \to 0$. Hence, we have completed the proof of Theorem \ref{thm:log}.
	
\end{proof}

\subsection{Proof of Proposition \ref{thm:collapse}}
\label{sec:proof_thm_collapse}
Before we present the proof, let us first introduce some useful results. First, let us define
\begin{align*}
  &V_{\perp}(Z(t)) \triangleq  d_{\mathcal{R}^{(r)}}(\Q(t))\\
  &V_{\perp l}(Z(t)) \triangleq d_{\mathcal{R}_l^{(r)}}(\Q(t)) \\
  &V_{\perp u}(Z(t)) \triangleq d_{\mathcal{R}_u^{(r)}}(\Q(t)).
\end{align*}
By Eq. \eqref{eq:distoregion}, we have $V_{\perp}(Z(t)) = \min\{V_{\perp l}(Z(t)), V_{\perp u}(Z(t))\}$. As a result, the drift of $V_{\perp}(Z)$ has the following four cases.

{Case 1:} $\Delta V_{\perp}(Z) = \Delta V_{\perp l}(Z)$

{Case 2:} $\Delta V_{\perp}(Z) = \Delta V_{\perp u}(Z)$

{Case 3:} $\Delta V_{\perp}(Z)= [V_{\perp l}(Z(t_0+1)) - V_{\perp u}(Z(t_0))] \mathcal{I}(Z(t_0) = Z)$

{Case 4:} $\Delta V_{\perp}(Z)= [V_{\perp u}(Z(t_0+1)) - V_{\perp l}(Z(t_0))] \mathcal{I}(Z(t_0) = Z)$.

Note that the drift in Case 3 can be upper bounded by $\Delta V_{\perp u}(Z)$ and the drift in Case 4 can be upper bounded by $\Delta V_{\perp l}(Z)$. Thus, in order to establish upper bounds on the drift of $V_{\perp}(Z)$, we only need to focus on the first two cases. \emph{In the following, we might omit the superscript $^{(r)}$ for ease of exposition, and revive it when necessary.}

Let us also define 
\begin{align*}
  \mathcal{R}_l^{\prime} \triangleq \mathcal{R}_l^{(r)} -  \mathbf{r} \text{ and } \mathcal{R}_u^{\prime} \triangleq \mathcal{R}_u^{(r)} - \mathbf{r}.
\end{align*}
where $\mathbf{r} = r\mathbf{1}$.
Correspondingly, we shift the queue-length vector in the same direction. That is, we let 
\begin{align}
\label{eq:def_Qprime}
  \mathbf{\Q}^{\prime} = \Q - \mathbf{r}.
\end{align}
The main motivation behind this shifting process is that it allows us to decompose queue-length vector into parallel and perpendicular components. In particular, given a queue length vector $\Q$, we have the following decompositions
\begin{align*}
  &\Q^{\prime} =  \Q_{\parallel \mathcal{R}_{l}^{\prime}}^{\prime} + \Q_{\perp \mathcal{R}_{l}^{\prime}}^{\prime}\\
  &\Q^{\prime} =  \Q_{\parallel \mathcal{R}_{u}^{\prime}}^{\prime} + \Q_{\perp \mathcal{R}_{u}^{\prime}}^{\prime},
\end{align*}
where $\Q_{\parallel \mathcal{R}_{l}^{\prime}}^{\prime}$ and $\Q_{\parallel \mathcal{R}_{u}^{\prime}}^{\prime}$ are the projections of $\Q^{\prime}$onto $\mathcal{R}_{l}^{\prime}$ and $\mathcal{R}_u^{\prime}$, referred as parallel components. $\Q_{\perp \mathcal{R}_{l}^{\prime}}^{\prime}$ and $\Q_{\perp \mathcal{R}_{u}^{\prime}}^{\prime}$ are the corresponding remainders, referred as perpendicular components. Note that the two decompositions are well defined and unique because $\mathcal{R}_{l}^{\prime}$ and $\mathcal{R}_u^{\prime}$ are both closed and convex. Moreover, we have 
\begin{align}
\label{eq:dis_relation}
  V_{\perp l}(Z(t)) = \big\lVert{\Q_{\perp \mathcal{R}_{l}^{\prime}}^{\prime}}\big\rVert \text{ and } V_{\perp u}(Z(t)) = \norms{\Q_{\perp \mathcal{R}_{u}^{\prime}}^{\prime}}.
\end{align}
This follows directly from the fact that the shifting process would not change the distance.

Now, we are ready to present our proof.
\begin{proof}[Proof of Proposition \ref{thm:collapse}]
   Since the chain $\{Z(t), t\ge 0\}$ is ergodic under JBT for any $r\ge 1$ by Lemma \ref{thm:throughput}, we can apply Lemma \ref{lem:basis} to establish bounded moments of $\overline{V}_{\perp}$. In particular, all we need to do is to check the drift conditions (C1) and (C2), respectively. As discussed above, we should only focus on the drifts $\Delta V_{\perp l}(Z)$ and $\Delta V_{\perp u}(Z)$.

   For condition (C2), we have the following result, the proof of which is relegated to Appendix \ref{sec:proof_claim_C2}.
   \begin{claim}
   	\label{claim_C2}
   	For any $t\ge0$, we have
   	\begin{align*}
            |\Delta V(Z(t))| \le \sqrt{N} \max(A_{{max} },S_{max}).
    \end{align*}
   \end{claim}
   This directly verifies condition (C2) in Lemma~\ref{lem:basis}.
    Now, we turn to check condition (C1) for $V_{\perp}(Z)$. To this end, we need the following result, the proof of which is relegated to Appendix \ref{sec:proof_of_claim_1}.

   \begin{claim}
    \label{claim_1}
    For any $t\ge0$, we have 
    \begin{align}
    \label{eq:drift_lower}
      &\ex{\Delta V_{\perp l}(Z) \mid Z(t) = Z }\nonumber\\
       \le & \frac{1}{2\lVert{\Q_{\perp \mathcal{R}_{l}^{\prime}}^{\prime}(t)}\rVert} \ex{\left(2\inner{{\Q_{\perp \mathcal{R}_{l}^{\prime}}^{\prime}}(t)}{\A(t) - \s(t)} + L\right) \mid Z(t) = Z}
    \end{align}
  and 
    \begin{align}
    \label{eq:drift_upper}
      &\ex{\Delta V_{\perp u}(Z) \mid Z(t) = Z }\nonumber\\
       \le & \frac{1}{2\lVert{\Q_{\perp \mathcal{R}_{u}^{\prime}}^{\prime}(t)}\rVert} \ex{\left(2\inner{{\Q_{\perp \mathcal{R}_{u}^{\prime}}^{\prime}}(t)}{\A(t) - \s(t)} + L\right) \mid Z(t) = Z}
    \end{align}
  where $L = N\max(A_{max},S_{max})^2$.
    \end{claim}
    From Claim \ref{claim_1}, we can see that the upper bounds on the mean drifts of $\Delta V_{\perp l}(Z)$ and $\Delta V_{\perp u}(Z)$ have the same formula. Thus, we can rewrite it in a compact way as follows.
    \begin{align}
    \label{eq:compact}
      &\ex{\Delta V_{\perp s}(Z) \mid Z(t) = Z }\nonumber\\
       \le & \frac{1}{2\lVert{\Q_{\perp \mathcal{R}_{s}^{\prime}}^{\prime}(t)}\rVert} \ex{\left(2\inner{{\Q_{\perp \mathcal{R}_{s}^{\prime}}^{\prime}}(t)}{\A(t) - \s(t)} + L\right) \mid Z(t) = Z}
    \end{align}
    where the subscript $s \in \{l, u\}$. To upper bound the right-hand side of Eq. \eqref{eq:compact}, we resort to the following result, the proof of which is relegated to Appendix~\ref{sec:proof_of_claim_3}.
    \begin{claim}
    \label{claim_3}
    For $s \in \{l,u\}$ and any system state $Z(t)$ with $V_{\perp}(Z(t)) > 0$, we have
    \begin{align*}
      \ex{\inner{{\Q_{\perp \mathcal{R}_{s}^{\prime}}^{\prime}}(t)}{\A(t) - \s(t)}  \mid Z(t) = Z}\le  - \frac{\mu_{\Sigma}\delta}{2N}\norms{{\Q_{\perp \mathcal{R}_{s}^{\prime}}^{\prime}}(t)},
    \end{align*}
    whenever $\epsilon \le \frac{\mu_{\Sigma}\delta}{2N+\delta}$, in which 
    \[\delta = \frac{\mu_{min}{\mu}_{min,2}}{\mu_{\Sigma}(\mu_{\Sigma}-\mu_{min})},\]
    where $\mu_{min} = \min_{n\in \mathcal{N}}\mu_n$, i.e., the smallest service rate among all servers. $\mu_{min,2}$ is the second smallest service rate among all the servers. Hence, $\delta$ is a constant independent of $\epsilon$. 
    \end{claim}

  Now substituting the upper bound in Claim~\ref{claim_3} into Eq. \eqref{eq:compact}, yields
    \begin{align*}
      &\ex{\Delta V_{\perp s}(Z) \mid Z(t) = Z }\nonumber\\
       \le & \frac{1}{2\lVert{\Q_{\perp \mathcal{R}_{s}^{\prime}}^{\prime}(t)}\rVert} \ex{\left(2\inner{{\Q_{\perp \mathcal{R}_{s}^{\prime}}^{\prime}}(t)}{\A(t) - \s(t)} + L\right) \mid Z(t) = Z}\\
       \le & - \frac{\mu_{\Sigma}\delta}{2N} + \frac{L}{2V_{\perp s}(Z)} \text{ whenever } \epsilon \le \frac{\mu_{\Sigma}\delta}{2N+\delta}\\
       \le & - \frac{\mu_{\Sigma}\delta}{4N},
    \end{align*}
    for $s \in \{l,u\}$ and for any $Z(t)$ such that $V_{\perp}(Z(t)) >0$ and $V_{\perp s}(Z(t)) \ge \frac{2NL}{\mu_{\Sigma}\delta}$. 

    Therefore, since the drift of $V_{\perp}(Z(t))$ is either upper bounded by the drift of $V_{\perp l}(Z(t))$ or the drift $V_{\perp u}(Z(t))$, and $V_{\perp}(Z(t))= \min\{V_{\perp l}(Z(t)), V_{\perp u}(Z(t)) \}$, we have 
    \begin{align*}
      \ex{\Delta V_{\perp }(Z) \mid Z(t) = Z } \le - \frac{\mu_{\Sigma}\delta}{4N} \text{ whenever } V_{\perp }(Z(t)) \ge \frac{2NL}{\mu_{\Sigma}\delta}
    \end{align*}
    for any $\epsilon \le \epsilon_0\triangleq \frac{\mu_{\Sigma}\delta}{2N+\delta}$. 

    Thus, condition (C1) in Lemma \ref{lem:basis} is validated with $\kappa = \frac{2NL}{\mu_{\Sigma}\delta}$ and $\eta = \frac{\mu_{\Sigma}\delta}{4N}$, both of which are independent of $\epsilon$ (since $\delta$ is independent of $\epsilon$ by Claim~\ref{claim_3}). Having established conditions (C1) and (C2) for the Lyapunov function $V_{\perp}(Z)$, by Lemma \ref{lem:basis}, we have that there exist some positive constants $\epsilon_0$, $\theta^*$ and $C^*$ such that for all $\epsilon \in (0, \epsilon_0)$
  \begin{align*}
    \ex{ e^{\theta^* d_{\mathcal{R}^{(r)}}\big(\overline{\Q}^{(\epsilon)}\big)} } \le C^*,
  \end{align*}
  where both $\theta^*$ and $C^*$ are independent of $\epsilon$. Hence, we have completed the proof of Proposition \ref{thm:collapse}.
\end{proof}
\section{Conclusion}

We have investigated the performance of load balancing systems under a general pull-based policy with a varying threshold. In particular, we have shown that a necessary condition for steady-state heavy-traffic delay optimality is that the threshold must grow to infinity as the load intensity approaches one but its growth rate should be slower than a certain polynomial function of the mean number of tasks in the system. We then showed that a sufficient condition to guarantee steady-state heavy-traffic delay optimality in pull-based load balancing systems is that the threshold must grow logarithmically with the mean number of tasks in the system, which directly resolves a generalized version of the conjecture by Kelly and Laws~\cite{kelly1993dynamic}. Both of the necessary and sufficient conditions are achieved by overcoming various technical challenges, and the methods developed in this paper could be of independent interest. In particular,
the methods developed in this paper might provide new directions on establishing steady-state delay optimality for dynamic threshold based scheduling policies in~\cite{harrison1998heavy,bell2001dynamic}. 

We finally conjecture that a logarithmic growth rate of the threshold is also necessary for heavy-traffic delay optimality in pull-based load balancing systems, and one possible future work is to extend the current proof of Theorem~\ref{thm:constant} to prove this result, hence providing a tighter characterization of general pull-based load balancing schemes in heavy traffic.

\bibliographystyle{plain}
\bibliography{ref} 

\section*{Appendix}
\appendix

\section{Proof of Lemma \ref{thm:throughput}}
\label{sec:appex_proof_throughput}
\begin{proof}
  To begin with, we first show that the Markov chain $\{Z(t) = (\Q(t),m(t)), t\ge 0\}$ is irreducible and aperiodic. Let the initial state be $Z(0) = (\Q(0), m(0)) = (0_{1 \times N}, m_0)$ where $m_0$ is the memory state in which all the $N$ IDs of servers are in the memory. The Markov chain is irreducible since for any state $Z$ in the state space, the Markov chain is able to reach the initial state within a finite step. This happens when there are no exogenous arrivals and all the offered service is at least one during each time-slot, which has a positive probability under our assumptions. The aperiodicity of the Markov chain $\{Z(t) = (\Q(t),m(t)), t\ge 0\}$ follows from the fact that the transition probability from the initial state to itself is positive.
  In order to show positive recurrence, we adopt the Foster-Lyapunov theorem. In particular, we only need to consider the Lyapunov function $W(Z) \triangleq \norm{\Q}^2$ since the memory state is finite. Now for any $t_0$, the one-step drift is given by 
  \begin{align}
  \label{eq:foster}
    &\ex{W(Z(t_0+1)) - W(Z(t_0)) \mid Z(t_0)}\nonumber\\
    = & \ex{\norm{\Q(t_0) + \A(t_0) - \s(t_0) + \UU(t_0)}^2 - \norm{\Q(t_0)}^2 \mid Z(t_0)}\nonumber\\
    \lep{a} & \ex{\norm{\Q(t_0) + \A(t_0) - \s(t_0)}^2 - \norm{\Q(t_0)}^2 \mid Z(t_0)}\nonumber\\
    = & \ex {2\inner{\Q(t_0)}{\A(t_0) - \s(t_0)} + \norm{\A(t_0) - \s(t_0)}^2  \mid Z(t_0) }\nonumber\\
    \lep{b} & \ex{2\inner{\Q(t_0)}{\A(t_0) - \s(t_0)} \mid Z(t_0)} + L\nonumber\\
    \lep{c} & 2 \sum_{n = 1}^N Q_n(t_0) \left(-\epsilon \frac{\mu_n}{\mu_{\Sigma}}\right) + L\nonumber\\
    \lep{d} & -2\epsilon\frac{\mu_{min}}{\mu_{\Sigma}}\norm{\Q(t_0)} + L,
  \end{align}
  where (a) follows from the facts that $Q_n(t) + A_n(t) -S_n(t) + U_n(t) = \max(Q_n(t) + A_n(t) -S_n(t),0)$ for any $t\ge 0$, and $\left(\max(a,0) \right)^2 \le a^2$ for any $a \in \mathbb{R}$; (b) holds since both the arrival and service processes have finite supports and $L = N \max(A_{max}, S_{max})^2$; (c) is true since under the JBT policy the worst case is when (proportionally) random routing is adopted, which happens if the ID in memory is either empty or full; (d) comes from the fact that $\norm{\mathbf{x}}_1 \ge \norm{\mathbf{x}}$ for any $\mathbf{x} \in \mathbb{R}^N$. Therefore, by the Foster-Lyapunov theorem, the Markov chain $\{Z(t) = (\Q(t),m(t)), t\ge 0\}$ is positive recurrent.

  Having established the fact that $\{Z(t) = (\Q(t),m(t)), t\ge 0\}$ is irreducible, aperiodic and positive recurrent, we are now ready to apply Lemma \ref{lem:basis} to show bounded moments of $\big\lVert \overline{\Q}\big\rVert$. Let us consider the Lyapunov function $V(Z) = \norm{\Q}$, and check the two conditions (C1) and (C2) in Lemma \ref{lem:basis}, respectively. 

  For condition (C1), we have
  \begin{align*}
    &\ex{ \Delta V(Z) \mid Z(t_0) = Z}\\
    = & \ex{\norm{\Q(t_0 + 1)} - \norm{\Q(t_0)}  \mid Z(t_0) = Z}\\
    = & \ex{\sqrt{\norm{\Q(t_0 + 1)}^2} - \sqrt{\norm{\Q(t_0)}^2} \mid Z(t_0) = Z }\\
    \lep{a} & \frac{1}{2\norm{\Q(t_0)}} \ex{\norm{\Q(t_0 + 1)}^2 - \norm{\Q(t_0)}^2 \mid Z(t_0) = Z}\\
    \lep{b} & -\epsilon\frac{\mu_{min}}{\mu_{\Sigma}} + \frac{L}{2\norm{\Q(t_0)}},
  \end{align*}
  where (a) follows from the fact that $f(x) = \sqrt{x}$ is concave; (b) comes from Eq. \eqref{eq:foster}. Thus, condition (C1) is valid with $\kappa = \frac{L\mu_{\Sigma}}{\epsilon \mu_{min}}$ and $\eta = \frac{\epsilon\mu_{min}}{2\mu_{\Sigma}}$.

  For condition (C2), we have 
  \begin{align*}
     |\Delta V(Z)| &= | \norm{\Q(t_0+1)} - \norm{\Q(t_0)} | \mathcal{I}(Z(t_0) = Z)\\
          & \lep{a} \norm{\Q(t_0+1) - \Q(t_0)}\mathcal{I}(Z(t_0) = Z)\\
          & \lep{b} \sqrt{N}\max(A_{max},S_{max}),
  \end{align*}
  where (a) holds since $| \norm{\mathbf{x}} - \norm{\mathbf{y}} | \le \norm{\mathbf{x} - \mathbf{y}}$ for each $\mathbf{x}$, $\mathbf{y}$ in $\mathbb{R}^N$; (b) follows from the assumptions that $A_{\Sigma}(t) \le A_{max}$ and $S_n(t) \le S_{max}$ for any $t\ge 0$ and $n \in \mathcal{N}$. Thus, condition (C2) is valid with $D = \sqrt{N}\max(A_{max},S_{max})$.

  Therefore, according to Eq. \eqref{eq:upper_siva} in Lemma \ref{lem:basis}, we get for $p = 1,2,\ldots,$
  \begin{align*}
    \ex{\big \lVert {\overline{\Q}^{(\epsilon)} } \big\rVert ^p } &\le \frac{1}{\epsilon^p}\left( \frac{2L\mu_{\Sigma}}{\mu_{min}}\right)^p + \frac{1}{\epsilon^p} \left(\frac{8D \mu_{\Sigma} }{\mu_{min}} \right)^p(D+\mu_{min})^p p!\\
    & \le  \frac{M_p}{\epsilon^p},
  \end{align*}
  where the constant $M_p = \left( \frac{2L\mu_{\Sigma}}{\mu_{min}}\right)^p + p! \left(\frac{8D \mu_{\Sigma} }{\mu_{min}} \right)^p(D+\mu_{min})^p$.

  In addition, if we apply Theorem 2.3 in~\cite{hajek1982hitting}, we can obtain that 
  \begin{align}
  \label{eq:mgf}  
    \ex{ e^{\theta^* \normss{\overline{\Q}^{(\epsilon)} } }} \le K_1 e^{\theta^* K_2/\epsilon},
  \end{align}
  where the positive constants $\theta^*$, $K_1$ and $K_2$ are all independent of $\epsilon$.
\end{proof}  

\section{Proof of Lemma \ref{lem:key_equation}}
\label{sec:proof_lemma5}
\begin{proof}

Let us consider the following Lyapunov function:
\begin{equation*}
	V_1(Z) \triangleq \sum_{i=1}^N\sum_{j>i}^N \left(Q_i-Q_j\right)^2.
\end{equation*}
We start with the conditional mean drift of $V_1(Z)$. Note that we shall omit the time reference $(t)$ after the first step and $\Q^+ \triangleq \Q(t+1)$.
\begin{equation*}
	\begin{split}
		&\ex{V_1(Z(t+1)) - V_1(Z(t)) \mid Z(t) = Z}\\
		= &  \sum_{i=1}^N\sum_{j>i}^N \ex{\left(Q_i(t+1)-Q_j(t+1)\right)^2 - \left(Q_i(t) - Q_j(t) \right)^2 \mid Z(t) = Z}\\
		= &  \sum_{i=1}^N\sum_{j>i}^N \ex{2 \left(Q_i - Q_j \right)\left(A_i - A_j - S_i + S_j \right)  - \left(U_i -U_j\right)^2\mid Z}\\
		& + \sum_{i=1}^N\sum_{j>i}^N \ex{\left(A_i - A_j - S_i + S_j\right)^2 + 2 \left(Q_i^+ - Q_j^+ \right)\left( U_i - U_j\right) \mid Z}\\
		\ep{a} & \sum_{i=1}^N\sum_{j>i}^N \ex{2 \left(Q_i - Q_j \right)\left(A_i - A_j \right)  - \left(U_i -U_j\right)^2\mid Z}\\
		& + \sum_{i=1}^N\sum_{j>i}^N \ex{\left(A_i - A_j - S_i + S_j\right)^2 - 2 \left(Q_i^+U_j + Q_j^+ U_i \right) \mid Z},\\
	\end{split}
\end{equation*}
in which (a) follows from the fact that the service is independent of queue lengths and homogeneous, as well as $Q_n(t+1)U_n(t) = 0$ for all $n$ and $t > 0$.

Since $\norm{\Q}$ has a finite second moment in steady state under JBT by Lemma~\ref{thm:throughput}, the steady-state mean $\ex{V_1(\overline{Z}^{(\epsilon)})}$ is finite for any $\epsilon > 0$. As a result, the mean drift of $V_1(\cdot)$ is zero in steady state, which directly implies the result in Lemma~\ref{lem:key_equation}.
\end{proof}
\section{Proof of Claim \ref{claim_C2}}
\label{sec:proof_claim_C2}
\begin{proof}
  For any $t_0 \ge 0$, we have
\begin{align*}
            & |\Delta V_{\perp l}(Z)| \nonumber\\
            \ep{a} & | \lVert{\Q_{\perp \mathcal{R}_{l}^{\prime}}^{\prime}(t_0+1)}\rVert - \lVert{\Q_{\perp \mathcal{R}_{l}^{\prime}}^{\prime}(t_0)}\rVert | \mathcal{I}(Z(t_0) = Z)\nonumber \\
            \lep{b} & \lVert{ {\Q_{\perp \mathcal{R}_{l}^{\prime}}^{\prime}(t_0+1)} - {\Q_{\perp \mathcal{R}_{l}^{\prime}}^{\prime}(t_0)}    }\rVert\mathcal{I}(Z(t_0) = Z)\nonumber\\
            \lep{c} & \lVert{\Q^{\prime}(t_0+1) - \Q^{\prime}(t_0)}\rVert \mathcal{I}(Z(t_0) = Z)\nonumber\\
            \ep{d} & \lVert{\Q(t_0+1) - \Q(t_0)}\rVert \mathcal{I}(Z(t_0) = Z)\nonumber\\
            \lep{e} & \sqrt{N} \max(A_{{max} },S_{max}),
    \end{align*}
    where (a) follows from Eq. \eqref{eq:dis_relation}; (b) comes from the fact that  $|\norm{{\bf x}} - \norm{{\bf y}}| \le \norm{{\bf x} - {\bf y}}$ holds for any ${\bf x}$, ${\bf y} \in \mathbb{R}^N$; (c) is due to the non-expansive property of projection and the fact that $\Q_{\perp \mathcal{R}_{l}^{\prime}}^{\prime}$ is the projection of $\Q^{\prime}$ onto the  polar cone of $\mathcal{R}_{l}^{\prime}$; (d) follows from the definition of $\Q^{\prime}$ in Eq. \eqref{eq:def_Qprime}; (e) holds due to the assumptions that the $A_\Sigma(t) \le A_{max}$ and $S_n(t) \le S_{max}$ for all $t \ge 0$ and all $ 1 \le n \le N$. With the same arguments, we can establish that 
    \begin{align*}
            & |\Delta V_{\perp u}(Z)| \le \sqrt{N} \max(A_{max},S_{max}).
    \end{align*}
    Since the drift of $V_{\perp}(Z)$ is either upper bounded by $\Delta V_{\perp l}(Z)$ or $\Delta V_{\perp u}(Z)$, we finally get 
    \begin{align*}
      |\Delta V_{\perp }(Z)| \le \sqrt{N} \max(A_{max},S_{max}).
    \end{align*}
    \end{proof}

\section{Proof of Claim \ref{claim_1}}
\label{sec:proof_of_claim_1}
\begin{proof}
  We first start with inequality \eqref{eq:drift_upper} in Claim \ref{claim_1}. Let us define 
    \begin{align*}
      &\Delta W(Z) = \left[\normss{\Q^{\prime}(t+1)}^2 - \normss{\Q^{\prime}(t)}^2 \right]\mathcal{I}(Z(t) = Z)\\
      &\Delta W_{\parallel u}(Z) = \left[\normss{\Q_{\parallel \mathcal{R}_{u}^{\prime}}^{\prime}(t+1)}^2 - \normss{\Q_{\parallel \mathcal{R}_{u}^{\prime}}^{\prime}(t)}^2 \right]\mathcal{I}(Z(t) = Z).
    \end{align*}.
    
    Then, the mean drift of $\Delta V_{\perp u}(Z)$ can be decomposed as follows.
  \begin{align}
  \label{eq:driftU_decom}
  &\ex{\Delta V_{\perp u}(Z) \mid Z(t) = Z}\nonumber \\
  \ep{a} &\ex{ \normss{\Q_{\perp \mathcal{R}_{u}^{\prime}}^{\prime}(t+1)} - \normss{\Q_{\perp \mathcal{R}_{u}^{\prime}}^{\prime}(t)} \mid Z(t) = Z}\nonumber \\
  = & \left[\sqrt{\normss{\Q_{\perp \mathcal{R}_{u}^{\prime}}^{\prime}(t+1)}^2} - \sqrt{\normss{\Q_{\perp \mathcal{R}_{u}^{\prime}}^{\prime}(t)}^2}  \right]\mathcal{I}(Z(t) = Z)\nonumber \\
  \lep{b} & \frac{1}{2\normss{\Q_{\perp \mathcal{R}_{u}^{\prime}}^{\prime}(t)}}\ex{ \normss{\Q_{\perp \mathcal{R}_{u}^{\prime}}^{\prime}(t+1)}^2 - \normss{\Q_{\perp \mathcal{R}_{u}^{\prime}}^{\prime}(t)}^2 \mid Z(t) = Z}\nonumber \\
  \ep{c} & \frac{1}{2\normss{\Q_{\perp \mathcal{R}_{u}^{\prime}}^{\prime}(t)}}\ex{ \Delta W(Z) - \Delta W_{\parallel u}(Z)\mid Z(t) = Z}
    \end{align}
    where (a) follows from Eq. \eqref{eq:dis_relation}; (b) holds due to the concavity of function $f(x) = \sqrt{x}$ for $x \ge 0$; (c) comes from the Pythagorean theorem. Next, we will bound each term in Eq. \eqref{eq:driftU_decom}, respectively. To begin with, we have an upper bound for the first term as follows.
    \begin{align}
    \label{eq:drift_W}
      &\ex{\Delta W(Z) \mid Z(t) = Z}\nonumber\\
      = & \ex{\normss{\Q^{\prime}(t+1)}^2 - \normss{\Q^{\prime}(t)}^2 \mid Z(t) = Z}\nonumber\\
      \ep{a} & \ex{\normss{\Q(t+1) - \mathbf{r}}^2 - \normss{\Q(t) - \mathbf{r}}^2 \mid Z(t) = Z}\nonumber\\
      = & \ex{\normss{\Q(t) + \A(t) - \s(t) + \UU(t) - \mathbf{r}}^2 - \normss{\Q(t) - \mathbf{r}}^2 \mid Z(t) = Z}\nonumber\\
      = & \ex{\normss{\Q(t) + \A(t) - \s(t) - \mathbf{r}}^2 - \normss{\Q(t) - \mathbf{r}}^2\mid Z(t) = Z}\nonumber\\
      & + \ex{ \normss{\UU(t)}^2 + 2\inner{\Q(t+1) - \mathbf{r}- \UU(t)}{\UU(t)}\mid Z(t) = Z}\nonumber\\
      \lep{b} & \ex{2\inner{\Q^{\prime}(t)}{\A(t) - \s(t)} + \normss{\A(t) -\s(t)}^2 - 2\inner{\mathbf{r}}{\UU(t)} \mid Z(t) = Z }\nonumber\\
      \lep{c} & \ex{2\inner{\Q^{\prime}(t)}{\A(t) - \s(t)}- 2\inner{\mathbf{r}}{\UU(t)} \mid Z(t) = Z} + L,
    \end{align}
  where (a) follows from Eq. \eqref{eq:def_Qprime}; (b) holds because of $\inner{\Q(t+1)}{\UU(t)} = 0$ and the dropping of $-\normss{\UU(t)}^2$; in (c), $L = N \max(A_{max}, S_{max})^2$, which is true since both the arrival and service processes have finite support.

  We now turn to provide a lower bound on the second term in Eq. \eqref{eq:driftU_decom} as follows.
  \begin{align}
  \label{eq:drift_Wp_U}
        &\ex{\Delta W_{\parallel u}(Z) \mid Z(t) =  Z}\nonumber\\ 
        = & \ex{\normss{\Q_{\parallel \mathcal{R}_{u}^{\prime}}^{\prime}(t+1)}^2 - \normss{\Q_{\parallel \mathcal{R}_{u}^{\prime}}^{\prime}(t)}^2 \mid Z(t) = Z}\nonumber\\
        = & \ex{ 2\inner{\Q_{\parallel \mathcal{R}_{u}^{\prime}}^{\prime}(t)}{\Q_{\parallel \mathcal{R}_{u}^{\prime}}^{\prime}(t+1) - \Q_{\parallel \mathcal{R}_{u}^{\prime}}^{\prime}(t)} \mid Z}\nonumber\\
        &+ \ex{\norm{\Q_{\parallel \mathcal{R}_{u}^{\prime}}^{\prime}(t+1) - \Q_{\parallel \mathcal{R}_{u}^{\prime}}^{\prime}(t)}^2\mid Z}\nonumber\\
        \ge & \ex{2\inner{\Q_{\parallel \mathcal{R}_{u}^{\prime}}^{\prime}(t)}{\Q_{\parallel \mathcal{R}_{u}^{\prime}}^{\prime}(t+1) - \Q_{\parallel \mathcal{R}_{u}^{\prime}}^{\prime}(t)} \mid Z}\nonumber\\
        = &2 \ex{\inner{\Q_{\parallel \mathcal{R}_{u}^{\prime}}^{\prime}(t)}{\Q^{\prime}(t+1) -\Q^{\prime}(t)}\mid Z} \nonumber\\
        &- 2\ex{\inner{\Q_{\parallel \mathcal{R}_{u}^{\prime}}^{\prime}(t)}{\Q_{\perp \mathcal{R}_{u}^{\prime}}^{\prime}(t+1) - \Q_{\perp \mathcal{R}_{u}^{\prime}}^{\prime}(t)} \mid Z}\nonumber\\
        \gep{a} & \ex{2\inner{\Q_{\parallel \mathcal{R}_{u}^{\prime}}^{\prime}(t)}{\Q^{\prime}(t+1) -\Q^{\prime}(t)}\mid Z}\nonumber\\
        \gep{b} & \ex{2 \inner{\Q_{\parallel \mathcal{R}_{u}^{\prime}}^{\prime}(t)}{\A(t)-\s(t)} \mid Z},
  \end{align}
  where (a) holds because $\inner{\Q_{\parallel \mathcal{R}_{u}^{\prime}}^{\prime}(t)}{\Q_{\perp \mathcal{R}_{u}^{\prime}}^{\prime}(t)} = 0$ and $\inner{\Q_{\perp \mathcal{R}_{u}^{\prime}}^{\prime}(t+1)}{\Q_{\parallel \mathcal{R}_{u}^{\prime}}^{\prime}(t)} \le 0$ since $\Q_{\perp \mathcal{R}_{u}^{\prime}}^{\prime}(t+1)$ is in the polar cone of $\mathcal{R}_{u}^{\prime}$; (b) follows from Eq. \eqref{eq:def_Qprime} and the fact that all the components of $\Q_{\parallel \mathcal{R}_{u}^{\prime}}^{\prime}(t)$ and $\UU(t)$ are nonnegative. Thus, substituting Eqs. \eqref{eq:drift_W} and \eqref{eq:drift_Wp_U} into Eq.~\eqref{eq:driftU_decom}, yields 
  \begin{align*}
    &\ex{\Delta V_{\perp l}(Z) \mid Z(t) = Z }\nonumber\\
       \le & \frac{1}{2\lVert{\Q_{\perp \mathcal{R}_{l}^{\prime}}^{\prime}(t)}\rVert} \ex{\left(2\inner{{\Q_{\perp \mathcal{R}_{l}^{\prime}}^{\prime}}(t)}{\A(t) - \s(t)} + L\right) - 2\inner{\mathbf{r}}{\UU(t)} \mid Z}\\
       \lep{a}&\frac{1}{2\lVert{\Q_{\perp \mathcal{R}_{l}^{\prime}}^{\prime}(t)}\rVert} \ex{\left(2\inner{{\Q_{\perp \mathcal{R}_{l}^{\prime}}^{\prime}}(t)}{\A(t) - \s(t)} + L\right) \mid Z}
  \end{align*}
  where (a) holds since all the components of $\mathbf{r}$ and $\UU(t)$ are nonnegative. Thus, we have the bound in Eq. \eqref{eq:drift_upper} of Claim \ref{claim_1}.

  Next, we turn to the bound in inequality \eqref{eq:drift_lower}. Let us define
  \begin{align*}
    \Delta W_{\parallel l}(Z) = \left[\normss{\Q_{\parallel \mathcal{R}_{l}^{\prime}}^{\prime}(t+1)}^2 - \normss{\Q_{\parallel \mathcal{R}_{l}^{\prime}}^{\prime}(t)}^2 \right]\mathcal{I}(Z(t) = Z).
  \end{align*}
  With the same arguments as in Eq. \eqref{eq:driftU_decom}, the mean drift of $\Delta V_{\perp l}(Z)$ can be decomposed into two terms.
  \begin{align}
  \label{eq:driftL_decom}
  &\ex{\Delta V_{\perp l}(Z) \mid Z(t) = Z}\nonumber \\
  = &\ex{ \normss{\Q_{\perp \mathcal{R}_{l}^{\prime}}^{\prime}(t+1)} - \normss{\Q_{\perp \mathcal{R}_{l}^{\prime}}^{\prime}(t)} \mid Z(t) = Z}\nonumber \\
  \le & \frac{1}{2\normss{\Q_{\perp \mathcal{R}_{l}^{\prime}}^{\prime}(t)}}\ex{ \Delta W(Z) - \Delta W_{\parallel l}(Z)\mid Z(t) = Z}.
    \end{align}
    The first term can be upper bounded as in Eq. \eqref{eq:drift_W}. The second term can be lower bonded in a similar way as in Eq. \eqref{eq:drift_Wp_U} except the last step.
    \begin{align}
    \label{eq:drift_Wp_L}
      &\ex{\Delta W_{\parallel l}(Z) \mid Z(t) =  Z}\nonumber\\ 
      \gep{a} & \ex{2\inner{\Q_{\parallel \mathcal{R}_{l}^{\prime}}^{\prime}(t)}{\Q^{\prime}(t+1) -\Q^{\prime}(t)}\mid Z}\nonumber\\
      \ep{b}& \ex{2\inner{\Q_{\parallel \mathcal{R}_{l}^{\prime}}^{\prime}(t)}{\A(t)-\s(t)+\UU(t)}\mid Z}\nonumber\\
      \gep{c}& \ex{2\inner{\Q_{\parallel \mathcal{R}_{l}^{\prime}}^{\prime}(t)}{\A(t)-\s(t)}- 2\inner{\mathbf{r}}{\UU(t)}\mid Z},
    \end{align}
    where (a) follows from the same arguments as in Eq. \eqref{eq:drift_Wp_U}; (b) comes from the definition of $\Q^{\prime}$ in Eq. \eqref{eq:def_Qprime}; (c) is true since any component of $\Q_{\parallel \mathcal{R}_{l}^{\prime}}^{\prime}(t)$ is greater or equal to $-r$ by the definition of $\mathcal{R}_{l}^{\prime}$. Thus, substituting Eqs. \eqref{eq:drift_W} and \eqref{eq:drift_Wp_L} into Eq. \eqref{eq:driftL_decom} yields the bound in Eq. \eqref{eq:drift_lower} of Claim \ref{claim_1}. Hence, we complete the proof of Claim \ref{claim_1}.
\end{proof}

\section{Proof of Claim \ref{claim_3}}
  \label{sec:proof_of_claim_3}  
  \begin{proof}
    In order to analyze the inner product in Eq. \eqref{eq:compact}, it is advantageous to reorder the queue-length vector $\Q(t)$. More precisely, let $\sigma_t(\cdot)$ be a permutation of $(1,2,\ldots,N)$ such that $\Q_{\sigma_t(1)}(t) \le \Q_{\sigma_t(2)}(t) \le \ldots \le \Q_{\sigma_t(N)}(t)$ and ties are broken randomly. We define the permutation vectors as follows
  \begin{align*}
    &\widehat{\Q}(t) \triangleq (Q_{\sigma_t(1)}(t), Q_{\sigma_t(2)}(t),\ldots, Q_{\sigma_t(N)}(t))\\
    &\widehat{\A}(t) \triangleq (A_{\sigma_t(1)}(t), A_{\sigma_t(2)}(t),\ldots, A_{\sigma_t(N)}(t))\\
    &\widehat{\s}(t) \triangleq (S_{\sigma_t(1)}(t), S_{\sigma_t(2)}(t),\ldots, S_{\sigma_t(N)}(t)).
  \end{align*}
  Let $p_n(t)$ be the probability that the new arrivals are dispatched to queue $n$ at time-slot $t$, and $\widehat{\mathbf{P}}(t) = (p_{\sigma_t(1)}(t),p_{\sigma_t(2)}(t), \ldots, p_{\sigma_t(N)}(t))$, i.e., the $i$-th component of $\widehat{\mathbf{P}}(t)$ is the probability of dispatching arrivals to the $i$-th shortest queue at time-slot $t$. We define 
  \begin{align}
  \label{eq:def_delta}
    \Delta(t) = \widehat{\mathbf{P}}(t) - \widehat{\mathbf{P}}_{\text{rand}}(t),
  \end{align}
  where $\widehat{\mathbf{P}}_{\text{rand}}(t)$ denotes the permutation of the dispatching distribution $\mathbf{p}(t)$ under proportionally random routing, i.e., the $i$-th component of $\widehat{\mathbf{P}}_{\text{rand}}(t)$ is $\mu_{\sigma_t(i)}/\mu_{\Sigma}$. 

  As before, we let $\widehat{\Q}^{\prime}(t) = \widehat{\Q}(t) - \mathbf{r}$. By the symmetry of $\mathcal{R}_s^{\prime}$ with respect to the line $\mathbf{1} = (1,1,\ldots,1)$, we have that the permutation of the perpendicular component $\Q_{\perp \mathcal{R}_{s}^{\prime}}^{\prime}(t)$ is equal to the perpendicular component of the permutation of ${\Q}^{\prime}(t)$, which is denoted by $\widehat{\Q}^{\prime}_{\perp s}(t)$. That is, $\widehat{\Q}^{\prime}_{\perp s}(t) = \widehat{\Q}^{\prime}(t) - \widehat{\Q}^{\prime}_{\parallel \mathcal{R}_s^{\prime}}(t)$ in which $\widehat{\Q}^{\prime}_{\parallel \mathcal{R}_s^{\prime}}(t)$ is the projection of the vector $\widehat{\Q}^{\prime}(t)$ onto $\mathcal{R}_s^{\prime}$ and $s \in \{l,u\}$.

  Based on the notions introduced above, the inner product in Eq. \eqref{eq:compact} can be rewritten as follows.
  \begin{align}
  \label{eq:inner_rewrite}
    &\ex{\inner{{\Q_{\perp \mathcal{R}_{s}^{\prime}}^{\prime}}(t)}{\A(t) - \s(t)}  \mid Z(t) = Z}\nonumber\\
    \ep{a} & \ex{\inner{{\widehat{\Q}_{\perp s}^{\prime}}(t)}{\widehat{\A}(t) - \widehat{\s}(t)}  \mid Z(t) = Z}\nonumber\\
    \ep{b} & \sum_{n=1}^N \widehat{Q}_{\perp s, n}^{\prime}(t)\left[\lambda_{\Sigma}\left(\Delta_n(t) + \frac{\mu_{\sigma_t(n)}}{\mu_{\Sigma}}\right)- \mu_{\sigma_t(n)} \right]\nonumber\\
    \ep{c} & \sum_{n=1}^N \widehat{Q}_{\perp s, n}^{\prime}(t)\Delta_n(t)\lambda_{\Sigma} + \sum_{n=1}^N \widehat{\Q}_{\perp s, n}^{\prime}(t)\left(-\epsilon \frac{\mu_{\sigma_t(n)}}{\mu_{\Sigma}}\right)\nonumber\\
    \le & \sum_{n=1}^N \widehat{Q}_{\perp s, n}^{\prime}(t)\Delta_n(t)\lambda_{\Sigma} + \epsilon \norms{ \widehat{\Q}_{\perp s}^{\prime}(t)   }_1,
  \end{align}
  where (a) follows from the fact inner product remains the same under permutation and the fact that the permutation of $\Q_{\perp \mathcal{R}_{s}^{\prime}}^{\prime}(t)$ is equal to $\widehat{\Q}^{\prime}_{\perp s}(t)$ as shown above; (b) holds due to the definition of $\Delta(t)$ and $\widehat{Q}_{\perp s, n}^{\prime}(t)$ is the $n$-th component of $\widehat{\Q}^{\prime}_{\perp s}(t)$; (c) simply follows from $\lambda_{\Sigma} = \mu_{\Sigma} - \epsilon$.

  In order to further analyze Eq. \eqref{eq:inner_rewrite}, we need the following results, which are proved at the end of this proof.
  \begin{claim}
    \label{claim_2}
    Regarding the vectors $\widehat{\Q}_{\perp s}^{\prime}(t)$ and $\Delta(t)$ in Eq. \eqref{eq:inner_rewrite}, we have the following properties for any system state $Z(t)$ such that $V_{\perp}(Z(t)) > 0$.
    \begin{enumerate}[(a)]
      \item The vector $\widehat{\Q}_{\perp s}^{\prime}(t)$ satisfies $ \widehat{Q}_{\perp s, 1}^{\prime}(t) \le \widehat{Q}_{\perp s, 2}^{\prime}(t)\le \ldots \le \widehat{Q}_{\perp s, N}^{\prime}(t)$ and $\widehat{Q}_{\perp s, 1}^{\prime}(t) \le 0$, $\widehat{Q}_{\perp s, N}^{\prime}(t) \ge 0$, where $s \in \{l,u\}$. More precisely, we have 
      \begin{align}
      &\widehat{Q}_{\perp l, 1}^{\prime}(t) = 0 \text{ and } \widehat{\Q}_{\perp l,N}^{\prime}(t) > 0\label{eq:L_twosides}\\
      &\widehat{Q}_{\perp u, 1}^{\prime}(t) < 0 \text{ and } \widehat{\Q}_{\perp u,N}^{\prime}(t) = 0\label{eq:U_twosides}.
      \end{align}

      \item The vector $\Delta(t)$ satisfies for some $k \in \{2,3,\ldots,N\}$
      \begin{align*}
        \Delta_n(t) \ge 0, n < k \text{ and } \Delta_n(t) \le 0, n\ge k
      \end{align*}
      and 
      \begin{align*}
        \min\left(|\Delta_1(t)|,|\Delta_N(t)|\right) \ge \delta,
      \end{align*}
      for some constant $\delta$ that is independent of $\epsilon$.
    \end{enumerate}
  \end{claim}
    Based on Claim \ref{claim_2}, we can bound the first term in Eq. \eqref{eq:inner_rewrite} for any system state $Z(t)$ such that $V_{\perp}(Z(t)) > 0$ as follows
    \begin{align}
    \label{eq:upper_on_inner}
      \sum_{n=1}^N \widehat{Q}_{\perp s, n}^{\prime}(t)\Delta_n(t)\lambda_{\Sigma} \le -\lambda_{\Sigma}\delta\left(|\widehat{Q}_{\perp s, 1}^{\prime}(t)| + |\widehat{Q}_{\perp s, N}^{\prime}(t)|\right).
    \end{align}
    This inequality can be verified as follows. Since $\Delta(t)$ satisfies the property (b) in Claim \ref{claim_2}, it can be constructed in the following way. To start with, all the $\Delta_n(t)$ is equal to $0$. Then, we decrease $\Delta_N(t)$ and increase $\Delta_1(t)$ by the same amount of $\delta$. After this process, the left-hand side of Eq. \eqref{eq:upper_on_inner} is equal to $\lambda_{\Sigma}(\delta \widehat{Q}_{\perp s, 1}^{\prime}(t) - \delta \widehat{Q}_{\perp s, N}^{\prime}(t))$, which is equivalent to the right-hand side of Eq. \eqref{eq:upper_on_inner} because of 
    $\widehat{Q}_{\perp s, 1}^{\prime}(t) \le 0$, $\widehat{Q}_{\perp s, N}^{\prime}(t) \ge 0$ in (a) of Claim \ref{claim_2}. Then, due to the first condition in (b) of Claim \ref{claim_2} and the fact that $\sum_{n = 1}^N \Delta_n(t) = 0$, any further construction (if necessary) for $\Delta(t)$ can only take the following way: it decreases some amount (say $\beta$) from $\Delta_i(t)$ where $i\ge k$, and then increase the same amount, i.e., $\beta$ for some $\Delta_j(t)$ where $j < k$. Through this process, the left-hand side of Eq. \eqref{eq:upper_on_inner} can only further decrease due to the monotone nondecreasing property of $\widehat{\Q}_{\perp s}^{\prime}(t)$ in (a) of Claim \ref{claim_2}. As a result, we have established the upper bound in Eq. \eqref{eq:upper_on_inner}.

    Next, we can further bound the right-hand side of Eq. \eqref{eq:upper_on_inner} in terms of $\normss{\widehat{\Q}_{\perp s}^{\prime}(t)}_1$. First, consider the case when $s = l$, we have 
    \begin{align}
    \label{eq:upper_on_inner_L}
      \sum_{n=1}^N \widehat{Q}_{\perp l, n}^{\prime}(t)\Delta_n(t)\lambda_{\Sigma} &\le -\lambda_{\Sigma}\delta\left(|\widehat{Q}_{\perp l, 1}^{\prime}(t)| + |\widehat{Q}_{\perp l, N}^{\prime}(t)|\right)\nonumber\\
      & \le -\lambda_{\Sigma}\delta|\widehat{Q}_{\perp l, N}^{\prime}(t)|\nonumber\\
      & \lep{a} \frac{-\lambda_{\Sigma}\delta}{N} \norms{\widehat{\Q}_{\perp l}^{\prime}(t)}_1
    \end{align}
    where (a) holds since $\normss{\widehat{\Q}_{\perp l}^{\prime}(t)}_1 \le N |\widehat{Q}_{\perp l, N}^{\prime}(t)|$ by the monotone nondecreasing property of $\widehat{\Q}_{\perp s}^{\prime}(t)$ and Eq. \eqref{eq:L_twosides} in (a) of Claim \ref{claim_2}. Similarly, when $s = u$, we have 
    \begin{align}
    \label{eq:upper_on_inner_U}
      \sum_{n=1}^N \widehat{Q}_{\perp u, n}^{\prime}(t)\Delta_n(t)\lambda_{\Sigma} &\le -\lambda_{\Sigma}\delta\left(|\widehat{Q}_{\perp u, 1}^{\prime}(t)| + |\widehat{Q}_{\perp u, N}^{\prime}(t)|\right)\nonumber\\
      & \le -\lambda_{\Sigma}\delta|\widehat{Q}_{\perp u, 1}^{\prime}(t)|\nonumber\\
      & \lep{a} \frac{-\lambda_{\Sigma}\delta}{N} \norms{\widehat{\Q}_{\perp u}^{\prime}(t)}_1
    \end{align}
    where (a) holds since $\normss{\widehat{\Q}_{\perp u}^{\prime}(t)}_1 \le N |\widehat{Q}_{\perp l, N}^{\prime}(t)|$ by the monotone nondecreasing property of $\widehat{\Q}_{\perp s}^{\prime}(t)$ and Eq. \eqref{eq:U_twosides} in (a) of Claim \ref{claim_2}. 

    Therefore, based on Eqs. \eqref{eq:upper_on_inner_L} and \eqref{eq:upper_on_inner_U}, the left-hand side of Eq. \eqref{eq:upper_on_inner} can be upper bounded in terms of $\normss{\widehat{\Q}_{\perp s}^{\prime}(t)}_1$ as follows.
    \begin{align}
    \label{eq:inner_finally}
      \sum_{n=1}^N \widehat{Q}_{\perp s, n}^{\prime}(t)\Delta_n(t)\lambda_{\Sigma} \le \frac{-\lambda_{\Sigma}\delta}{N} \norms{\widehat{\Q}_{\perp s}^{\prime}(t)}_1
    \end{align}
    for $s \in \{l,u\}$ and any system state $Z(t)$ with $V_{\perp}(Z(t)) > 0$. 
    Now, substituting Eq. \eqref{eq:inner_finally} into Eq. \eqref{eq:inner_rewrite}, yields
    \begin{align*}
      &\ex{\inner{{\Q_{\perp \mathcal{R}_{s}^{\prime}}^{\prime}}(t)}{\A(t) - \s(t)}  \mid Z(t) = Z}\nonumber\\
      \le & \left(\epsilon - \frac{\lambda_{\Sigma}\delta}{N}\right)\norms{\widehat{\Q}_{\perp s}^{\prime}(t)}_1\nonumber\\
      \le & - \frac{\mu_{\Sigma}\delta}{2N} \norms{\widehat{\Q}_{\perp s}^{\prime}(t)}_1 \text{ whenever } \epsilon \le \frac{\mu_{\Sigma}\delta}{2N+\delta}\nonumber\\
      \le & - \frac{\mu_{\Sigma}\delta}{2N}\norms{{\Q_{\perp \mathcal{R}_{s}^{\prime}}^{\prime}}(t)},
    \end{align*}
    for $s \in \{l,u\}$ and any system state $Z(t)$ with $V_{\perp}(Z(t)) > 0$, in which the last inequality follows from the fact $\normss{{\Q_{\perp \mathcal{R}_{s}^{\prime}}^{\prime}}(t)}_1 = \normss{\widehat{\Q}_{\perp s}^{\prime}(t)}_1$ and $\norm{\mathbf{x}}_1 \ge \norm{\mathbf{x}}$ for any $\mathbf{x} \in \mathbb{R}^N$. Hence, we establish the result in Claim~\ref{claim_3}.

     Now, we give the proof of Claim \ref{claim_2}.

    For (a), by the definition of $\widehat{\Q}^{\prime}(t)$, we have $\widehat{Q}^{\prime}_1(t) \le \widehat{Q}^{\prime}_2(t) \le \ldots \le \widehat{Q}^{\prime}_N(t)$. The projection of $\widehat{\Q}^{\prime}(t)$ onto $\mathcal{R}_u^{\prime}$, which is equal to $\widehat{\Q}_{\perp l}^{\prime}(t)$, is given by 
    \begin{align}
    \label{eq:Qprime_l_max}
      \widehat{\Q}_{\perp l}^{\prime}(t) = \widehat{\Q}_{\parallel u}^{\prime}(t) = \max\left(\widehat{\Q}^{\prime}(t), \mathbf{0} \right).
    \end{align}
    As a result, we have 
    \begin{align}
    \label{eq:Qprime_u_min}
      \widehat{\Q}_{\perp u}^{\prime}(t) = \widehat{\Q}^{\prime}(t) - \widehat{\Q}_{\parallel u}^{\prime}(t) = \min\left(\widehat{\Q}^{\prime}(t), \mathbf{0} \right).
    \end{align}
    Therefore, we have $\widehat{Q}_{\perp s, 1}^{\prime}(t) \le \widehat{Q}_{\perp s, 2}^{\prime}(t)\le \ldots \le \widehat{Q}_{\perp s, N}^{\prime}(t)$ for $s \in \{l,u\}$.
  Moreover, since $V_{\perp}(Z(t)) > 0$, we have $\Q(t) \notin \mathcal{R}^{(r)}$, which implies that $\Q^{\prime}(t) \notin \mathcal{R}^{\prime}_l$ and  $\Q^{\prime}(t) \notin \mathcal{R}^{\prime}_u$. Thus, we have there exist queues $i$ and $j$ such that $Q^{\prime}_i(t) < 0$ and $Q^{\prime}_j(t) > 0$, which further gives $\widehat{Q}^{\prime}_1(t) < 0 $ and $\widehat{Q}^{\prime}_N(t) > 0 $. As a result, by Eqs. \eqref{eq:Qprime_l_max} and \eqref{eq:Qprime_u_min}, we have 
  \begin{align*}
    &\widehat{Q}_{\perp l, 1}^{\prime}(t) = 0 \text{ and } \widehat{\Q}_{\perp l,N}^{\prime}(t) > 0\\
    &\widehat{Q}_{\perp u, 1}^{\prime}(t) < 0 \text{ and } \widehat{\Q}_{\perp u,N}^{\prime}(t) = 0,
  \end{align*}
  which establishes $\widehat{Q}_{\perp s, 1}^{\prime}(t) \le 0$ and $\widehat{Q}_{\perp s, N}^{\prime}(t) \ge 0$, where $s \in \{l,u\}$. Hence, we have completed the proof of (a) in Claim \ref{claim_2}.

  Now let us consider (b) in Claim \ref{claim_2}. First, since $V_{\perp}(Z(t)) > 0$, we have $\Q(t) \notin \mathcal{R}^{(r)}$, which implies that there exists queues $i$ and $j$ such that $Q_i(t) < r$ and $Q_j(t) > r$. This means that the number of IDs in memory denoted by $|m(t)|$ is between $1$ and $N-1$. Suppose $|m(t)| = M \in \{1,2,\ldots,N-1\}$, then we have 
  \begin{align*}
    \Delta_n(t) > 0, n < k \text{ and } \Delta_n(t) < 0, n\ge k,
  \end{align*}
  where $k = M+1$. This is because for $n < k$
  \begin{align*}
    \Delta_n(t) \ep{a} \frac{\mu_{\sigma_t(n)}}{\sum_{i=1}^M\mu_{\sigma_t(i)}} - \frac{\mu_{\sigma_t(n)}}{\mu_{\Sigma}} \gp{b} 0,
  \end{align*}
  and for $n \ge k$
  \begin{align*}
    \Delta_n(t) \ep{c} 0 - \frac{\mu_{\sigma_t(n)}}{\mu_{\Sigma}} < 0,
  \end{align*}
  where (a) and (c) follow from the definition of $\Delta(t)$ in Eq. \eqref{eq:def_delta} and the JBT policy; (b) holds due to $\mu_{\Sigma} = \sum_{i=1}^N \mu_{\sigma_t(i)}$ and $M < N$. Moreover, with simple calculations, we get 
  \begin{align*}
    \min\left( |\Delta_1(t)|, |\Delta_N(t)| \right) \ge \frac{\mu_{min}{\mu}_{min,2}}{\mu_{\Sigma}(\mu_{\Sigma}-\mu_{min})},
  \end{align*}
  where $\mu_{min} = \min_{n\in \mathcal{N}}\mu_n$, i.e., the smallest service rate among all servers. $\mu_{min,2}$ is the second smallest service rate among all the servers. Hence, we complete the proof of Claim \ref{claim_2}.
\end{proof}

\end{document}